\documentclass{article}

\usepackage{geometry}
\usepackage{amsmath,graphicx,amssymb}
\usepackage[utf8]{inputenc}
\usepackage{amscd,amsfonts,amsthm,xcolor,authblk}
\usepackage[colorlinks=true]{hyperref}

\usepackage{tikz}

\usepackage{caption}
\usepackage{subcaption}

\newcommand{\bs}{\boldsymbol}

\newcommand{\mhat}{\hat{m}}

\newtheorem{theo}{Theorem}
\newtheorem{heur}{Heuristics}
\newtheorem{prop}[theo]{Proposition}

\theoremstyle{definition}
\newtheorem{defi}{Definition}

\newtheorem{rem}{Remark}

\begin{document}

\title{Impact of a block structure on the Lotka-Volterra model}

%
\author{Maxime Clenet$^{(1,2)}$, François Massol$^{(3)}$, Jamal Najim$^{(1)}$\thanks{Supported by CNRS Project 80 Prime | KARATE.}}
\affil{{\small (1)} CNRS and Université Gustave Eiffel, France\\{\small (2)} Université de Sherbrooke, Canada\\ {\small (3)} CNRS, Université de Lille, INSERM, CHU, Institut Pasteur Lille}
%
%
\maketitle
\begin{abstract}
The Lotka-Volterra (LV) model is a simple, robust, and versatile model used to describe large interacting systems such as food webs or microbiomes. The model consists of $n$ coupled differential equations linking the abundances of $n$ different species.
We consider a large random interaction matrix with independent entries and a block variance profile. The $i$th diagonal block represents the intra-community interaction in community $i$, while the off-diagonal blocks represent the inter-community interactions. The variance remains constant within each block, but may vary across blocks.

We investigate the important case of two communities of interacting species, study how interactions affect their respective equilibrium. We also describe equilibrium with feasibility (i.e., whether there exists an equilibrium with all species at non-zero abundances) and the existence of an attrition phenomenon (some species may vanish) within each community.

Information about the general case of $b$ communities ($b> 2$) is provided in the appendix.

\end{abstract}
\begin{raggedleft}
\end{raggedleft}

\begin{tikzpicture}[remember picture,overlay]
\hypersetup{hidelinks}
\node[anchor=north west,yshift=450pt,xshift=-70pt]
{ \href{https://doi.org/10.24072/pci.mcb.100235}{\includegraphics[height=35mm]{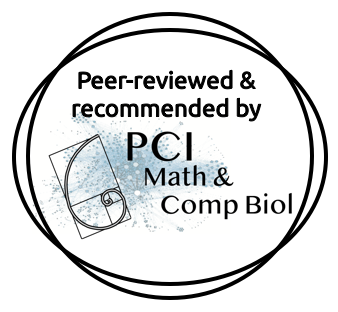}} } ;
\end{tikzpicture}

\newpage

\section{Introduction}

\paragraph{Motivations.}

Understanding large ecosystems and the underlying mechanisms that support high species diversity is a major challenge in theoretical ecology. Since the 1950's, understanding how species can stably coexist has been the focus of both theoretical \cite{macarthur_fluctuations_1955,Margalef1963,Levin1970,Roberts1974,hanski_coexistence_1983} and empirical studies \cite{Soliveres2015,saavedra_nested_2016,Leibold2017}. Motivated by the seminal work of May \cite{may_will_1972}, the introduction of random matrices has been a key mathematical step in modeling high-dimensional ecosystems \cite{allesina_stability_2012,stone_feasibility_2018,akjouj_complex_2022}. These tools have expanded our ability to understand the nature of interactions and how food webs can recover after small perturbations (stability) \cite{allesina_stability_2012,tang_correlation_2014}. In the course of the debate on the theory of species stable coexistence, a number of questions have emerged, including the following: What are the conditions which enable many species to coexist, especially regarding the structure of their interaction matrix? 

Differential equations are frequently used in ecology to describe a system of interacting species. One of the most common models is the Lotka-Volterra (LV) model \cite{lotka_elements_1925,volterra_fluctuations_1926}, which has been the subject of research in both ecology \cite{wangersky_lotka-volterra_1978, jansen_permanence_1987,law_self-assembling_1992} and mathematics \cite{goh_feasibility_1977,goh_global_1977,taylor_consistent_1988,hofbauer_evolutionary_1998,takeuchi_global_1996}. Certain properties of this model, such as its stability \cite{gibbs_effect_2018}, have raised much interest. The conditions under which all species survive, referred to as feasibility, have also motivated many works \cite{bizeul_positive_2021,grilli_feasibility_2017,stone_feasibility_2018}.

In nature, ecological networks are rather structured, and many studies have investigated the network structures that contribute to the stability of a given community \cite{thebault_stability_2010,allesina_predicting_2015}. 
One common network structure is food web compartmentalization, also known as modularity. The underlying concept is that the network is structured in the form of groups of nodes that interact more strongly within their group and more weakly between groups. A mathematical formulation of modularity was defined by Newman \cite{newman_modularity_2006}. Subsequently, modularity has been of great importance in ecology \cite{guimera_origin_2010}, in complex networks \cite{variano_networks_2004}, and in community detection (for a complete review, see Fortunato \cite{fortunato_community_2010}). 

May had already mentioned that a multi-community structure should improve stability \cite{may_will_1972}, a hypothesis later investigated by Pimm \cite{pimm_structure_1979}. In the same framework as May, Grilli \textit{et al.} \cite{grilli_modularity_2016} studied the effect of modularity on the stability of the Jacobian of a system, the so-called ``community matrix". However, studies show that modularity improves the persistence (:= non-extinction of species, generally related to their resistance to external perturbations) of species in the dynamical system \cite{stouffer_compartmentalization_2011}.

In this article, we study the Lotka-Volterra model where we consider a block structure network representing the inter- and intra-community interactions. Of particular interest are the interactions between the communities that affect their respective equilibrium and stability.

Each block is identified by its interaction strength, which is the standard deviation of the random part of the interactions. The idea that interaction strength plays a key role in the stability of ecosystems was introduced by May \cite{may_will_1972}. 
For the sake of mathematical simplicity, we limit our model to two communities, although we can extend the model to more complex food webs and multi-community frameworks. 

We study the existence and stability of an equilibrium, together with its properties. When an attrition phenomenon occurs (some species may vanish), we describe the proportion of surviving species and their distribution. We also provide conditions for which the equilibrium is feasible (i.e., whether there exists an equilibrium with all species at non-zero abundances).  

\paragraph{Known results.}

Building upon the insights gained from the model of May \cite{may_will_1972}, an understanding of the Lotka-Volterra model provides a foundation for in-depth analysis of the impact of interactions on community dynamics. Scientists from diverse disciplinary backgrounds, including mathematics, physics, and ecology, sought to investigate the intricacies of this complex ordinary differential equations system. The study of the stability of the Lotka-Volterra model has constituted a central focus of research, as evidenced by the works of Stone \cite{stone_feasibility_2018} and Gibbs \textit{et al.} \cite{gibbs_effect_2018}, which have been complemented by Clenet \textit{et al.} \cite{clenet_equilibrium_2023} where they investigated the properties of a stable equilibrium.

In fact, beyond the stability of the equilibria, the properties of these equilibria have been the subject of central interest \cite{servan_coexistence_2018,pettersson_predicting_2020, pettersson_stability_2020}, e.g. deriving the number of surviving species. Moreover, the existence of a feasible equilibrium and its stability have been demonstrated by Bizeul \textit{et al.} \cite{bizeul_positive_2021} where they establish that a threshold of interaction strength exists beyond which equilibrium of the system is almost certainly feasible. The methodology was further refined in the case of sparse interactions \cite{akjouj_feasibility_2022} and a correlation profile \cite{clenet_equilibrium_2022}. In order to gain a more comprehensive understanding of the LV model, Akjouj \textit{et al.} \cite{akjouj_complex_2022} conducted a comprehensive mathematical review of the subject.

Lotka-Volterra model provides an interesting diversity of dynamical behaviors, with partial mathematical knowledge. This is supplemented by methods from physics to improve the understanding on these various dynamical behaviors (properties of the equilibrium, out-of equilibrium dynamics, model sophistication). Bunin \cite{bunin_interaction_2016,bunin_ecological_2017} used the cavity methods to derive the properties of the surviving species and the multiple attractors phase. Barbier \textit{et al.} \cite{barbier_generic_2018} exhibits generic behaviors in complex communities.

Generating functional techniques for deriving similar mean-field equations to study the equilibrium phase in the LV system was used by Galla \cite{galla_dynamically_2018} and extended by Poley \textit{et al.} \cite{poley_generalized_2023} to study the LV model in the case of a cascade interaction matrix.

\paragraph*{Model and assumptions.}
The LV model is a standard model in ecology to study the dynamics of a community of species over time. It is defined by a system of $n$ differential equations
\begin{equation}\label{eq:LV}
\frac{dx_k}{dt}(t) = x_k(t)\, \left( r_k - x_k(t) + \sum_{\ell\in [n]} B_{k\ell} x_{\ell}(t)\right)\ ,\qquad k\in [n]=\{1,\cdots, n\}\ .
\end{equation}
The abundance of species $k$ at time $t$ is represented by $x_k(t)$ with $\bs{x} = (x_1,\cdots,x_n)$ the vector of abundances. 
Parameter $r_k$ corresponds to the growth rate of species $k$. The intraspecific parameter has been set to $1$ in accordance with the computations that were conducted in order to obtain a dimensionless LV model (for further details, please refer to Remark 2.1 in Akjouj \textit{et al.} \cite{akjouj_complex_2022}). The coefficient $B_{k\ell}$ represents the impact of species $\ell$ on species $k$. The $n\times n$ matrix $B=(B_{k\ell})$, which represents the interaction network, is decomposed into a block structure. This structure differentiates various groups of species in the form of communities that interact with each other. On the one hand, the diagonal blocks of $B$ correspond to interactions within each community, each with its own interaction strength. On the other hand, the off-diagonal blocks correspond to the impact of the communities on each other. Analytically and within the framework of two communities, the matrix $B = (B_{k\ell})_{n,n}$ is defined by blocks using random matrices $A_{ij}$ and interaction strengths $s_{ij}$:
\begin{equation}
\label{eq:mat_int}
B = \frac{1}{\sqrt{n}}\begin{pmatrix}
s_{11}A_{11} & s_{12}A_{12} \\ 
s_{21}A_{21} & s_{22}A_{22}
\end{pmatrix}\, ,
\end{equation}
where the random matrix $A_{ij}$ is non-Hermitian of size $|\mathcal{I}_i|\times|\mathcal{I}_j|$ with standard Gaussian entries i.e. $\mathcal{N}(0,1)$. Here $\mathcal{I}_1 = [n_1]$ (resp. ${\mathcal I}_2=\{n_1 +1,\cdots, n\}$), the subset of $[n]$ of size $|\mathcal{I}_1|=n_1$ (resp $|\mathcal{I}_2|=n_2$ - here and below $n=n_1+n_2$) matching the index of species belonging to community 1 (resp community 2). We define 
\begin{equation}
\label{def:beta}
\bs{\beta} = (\beta_1,\beta_2)\qquad \textrm{where}\qquad \beta_1 = \frac{n_1}n \quad \textrm{and}\quad \beta_2 = 1-\beta_1 = \frac{ n_2}n\,. 
\end{equation}
The Gaussianity assumption simplifies the explanations, but can be relaxed under certain circumstances (see the series of remarks, specifically Remarks \ref{rem:relaxed_assumptions_part1}, \ref{rem:relaxed_assumptions_part2} and \ref{rem:relaxed_assumptions_part3}, which can be found in their respective sections for further details). 

Notice a normalization parameter $1/\sqrt{n}$ in the matrix $B$. This enables the interaction matrix $B$ to have a {\it macroscopic effect} on system \eqref{eq:LV}. By macroscopic effect, we mean that even if the number of species $n$ grows to infinity, the effect of matrix $B$ in Eq. \eqref{eq:LV} remains noticeable (it does not vanish, nor does it explode). This can be illustrated by the following asymptotic properties of $B$ (hereafter $\|\cdot\|$ stands for the spectral norm):
$$
\| B\| \sim O(1)\ ;\qquad \mathbb{E}\left(\sum_{\ell\in [n]} B_{k\ell} x_{\ell}(t) \right) \sim O(1) \ ; \qquad  \mathrm{Var}\left(\sum_{\ell\in [n]} B_{k\ell} x_{\ell}(t) \right) \sim O(1) 
$$
as $n\to \infty$. These approximations are highly non trivial. The first one for instance is a generalization of the evaluation of the singular value of a matrix with i.i.d. entries and can be found in \cite{ajanki_quadratic_2019}. From an ecological perspective, this normalization has the following consequence: 
an increase in the number of species does not yield a corresponding increase in the overall strength of interactions between one species and all others, which order of magnitude remains similar.


The relative strength of interactions within and between blocks is controlled by the four $s_{ij}$ coefficients, which can be grouped together in a matrix $\bs{s}$:
$$
\bs{s} = \begin{pmatrix}
s_{11} & s_{12} \\ 
s_{21} & s_{22}
\end{pmatrix} \, .
$$
The diagonal terms $(s_{11},s_{22})$ represent the interaction strength in each community. The off-diagonal term $s_{12}$ (resp. $s_{21}$) represents the interaction strength of the impact of community $2$ on community $1$ (resp. community $1$ on community $2$). The lower the value of $s_{ij}$, the lower the rates of interaction between species.
Note that in the case of a unique community, $s$ is the interaction strength coefficient, i.e. the standard deviation of the interspecific coefficients of the LV model.

\begin{rem}
For the sake of simplicity, the results are presented in the case of two interacting communities but can be extended to the case of $b$ communities, see Appendix \ref{app:extension_b_block}.
\end{rem}

There are two scenarios of interest: Let us consider two separate groups of species that follow the dynamics described in model \eqref{eq:LV}. The matrices $A_{ij}$ are each sampled once. In the first scenario, we consider a very weak interaction between the two communities
$$
\bs{s} = \begin{pmatrix}
1/2 & \varepsilon\\ 
\varepsilon & 1/2
\end{pmatrix}, \, \varepsilon > 0 \, ,
$$
and the intra-communities interactions were selected to be relatively small (Fig. \ref{subfig:scenario1_matrix}) in order to observe that both communities' dynamics converge to a feasible equilibrium, in which all species survive (Fig. \ref{subfig:scenario1_dyn}). In the second scenario, we increase the interactions between the communities (Fig. \ref{subfig:scenario2_matrix}), i.e. the standard deviation matrix is defined by
$$
\bs{s} = \begin{pmatrix}
1/2 & 1\\ 
1 & 1/2
\end{pmatrix} \,  .
$$
In the case of a given inter-communities interaction realization, such as the one shown in Fig. \ref{fig:scenario2}, it is no longer possible for both communities to maintain the feasibility of all species. Some species are likely to disappear (Fig. \ref{subfig:scenario2_dyn}).
\begin{figure}[ht]
\centering
\begin{subfigure}[b]{0.43\textwidth}
  \centering
  \includegraphics[width=\textwidth]{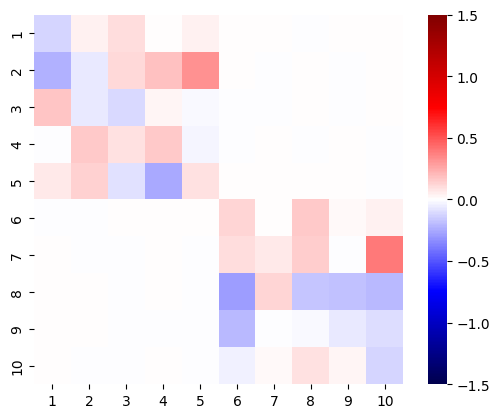}
  \caption{Interaction matrix}
  \label{subfig:scenario1_matrix}
\end{subfigure}
\hfill
\begin{subfigure}[b]{0.55\textwidth}
  \centering
  \includegraphics[width=\textwidth]{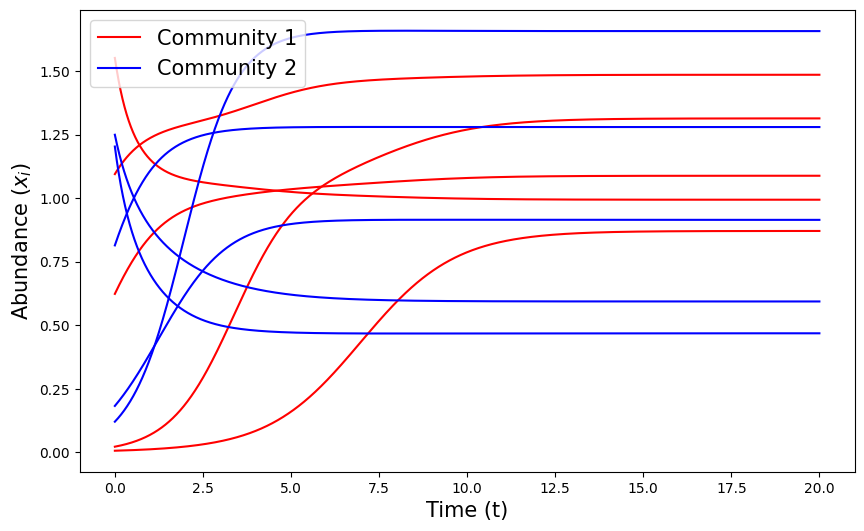}
  \caption{Abundance dynamics of two communities of five species}
  \label{subfig:scenario1_dyn}
\end{subfigure}
\caption{Dynamics of the model \eqref{eq:LV} with 2 distinct communities of 5 species each and interaction matrix given by \eqref{eq:mat_int}. The two communities converge to their feasible equilibrium point and do not interact. In Fig. (a), a heat map illustrates the interaction matrix \eqref{eq:mat_int}. Figure (b) shows the dynamics where the species of each community reach a feasible equilibrium.}
\label{fig:scenario1}
\end{figure}

\begin{figure}[ht]
\centering
\begin{subfigure}[b]{0.43\textwidth}
  \centering
  \includegraphics[width=\textwidth]{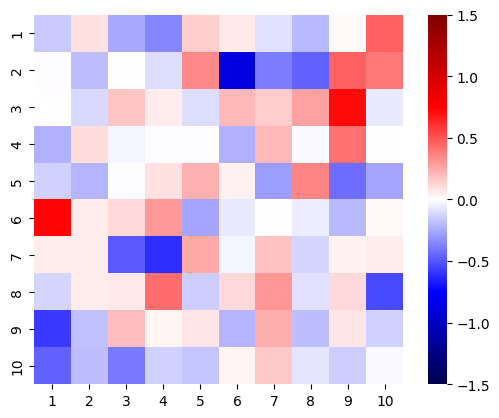}
  \caption{Interaction matrix}
  \label{subfig:scenario2_matrix}
\end{subfigure}
\hfill
\begin{subfigure}[b]{0.55\textwidth}
  \centering
  \includegraphics[width=\textwidth]{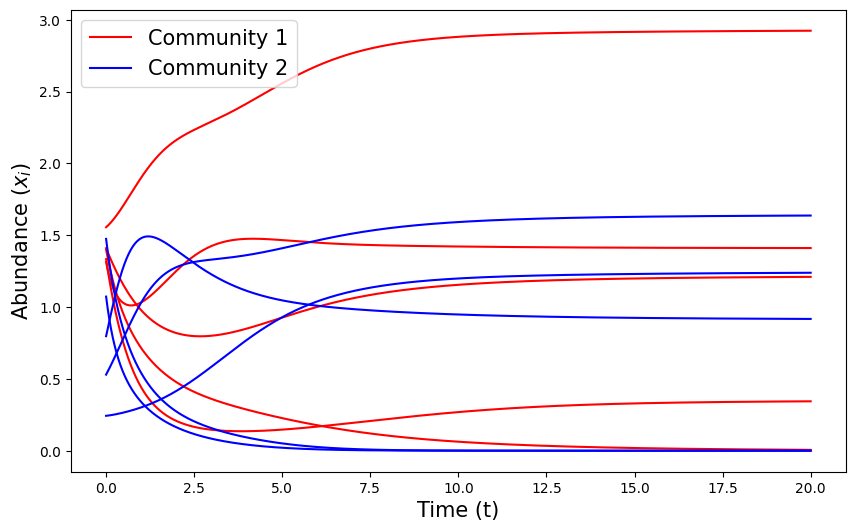}
  \caption{Abundance dynamics of two communities of five species}
    \label{subfig:scenario2_dyn}
\end{subfigure}
\caption{Dynamics of Model \eqref{eq:LV} with 2 distinct communities of 5 species each with interaction matrix given by \eqref{eq:mat_int}. In Fig. (a), representation of the interaction matrix \eqref{eq:mat_int} when the interactions between the communities are strong. Fig. (b) shows the dynamics where the species of each community reach an equilibrium. Notice that there are vanishing species in each community.}
\label{fig:scenario2}
\end{figure}

\paragraph*{Properties of the dynamical system.}
We are interested in the effect of a block structure on the food web, limit our study to the $2$-blocks case \eqref{eq:mat_int} and focus on the model with constant growth rate\footnote{The simplifying assumption $r_k = 1$ allows tractable computations and could be extended to $r_k = c$ with $c > 0$. However, if the growth rate is different for each species, the mathematical development and result may be strongly affected and will be discussed in a series of remarks, see Remarks \ref{rem:relaxed_assumptions_part1}, \ref{rem:relaxed_assumptions_part2} and \ref{rem:relaxed_assumptions_part3}.} $r_k = 1$:
\begin{equation}\label{eq:LV2}
\frac{dx_k}{dt} = x_k\, \left( 1 -  x_k + (B\boldsymbol{x})_k \right)\ ,\quad k\in [n] \, .
\end{equation}
Of major interest is the existence and uniqueness of an equilibrium $\bs{x^*}=(x^*_k)_{k\in[n]}$. The LV system is an autonomous differential system. If the initial conditions are positive i.e. $\bs{x}(0)>0$ (componentwise), it implies $\bs{x}(t)>0$ for every $t>0$. However, some of the components $x_k(t)$ may converge to zero if the equilibrium $\bs{x}^*$ has components equal to zero. 
An equilibrium to the LV system should hence satisfy the following set of constraints:
\begin{equation}\label{eq:equilibrium}
\begin{cases}
 & x^*_k\left( 1 -  x^*_k + (B\boldsymbol{x^*})_k \right) =0\,,\quad \forall k\in [n]\, ,\\ 
 & x^*_k \geq 0\,.
\end{cases}
\end{equation}
Two substantially different situations arise, that we will study hereafter.  

First, if $\bs{x}^*$ has vanishing components, the equilibrium equations are cast into a nonlinear optimization problem, which has been studied by Clenet \textit{et al.} \cite{clenet_equilibrium_2023} in the case of a single community.

If the equilibrium is feasible, that is $\bs{x}^*>0$, then the equilibrium set of equations becomes a linear equation: 
\begin{equation}
\label{eq:linear_feasibility}
\bs{x}^* = \bs{1}+B  \bs{x}^*\, .
\end{equation}
In the context of a single community, the existence of a positive solution has been studied by Bizeul and Najim \cite{bizeul_positive_2021} and extended for more complex food webs in \cite{akjouj_feasibility_2022,clenet_equilibrium_2022, liu_feasibility_2023}.

A further consideration which will be addressed is whether the equilibrium $\bs{x^*}$ is asymptotically globally stable, i.e. if for every initial vector $\bs{x}_0>0$ the solution of \eqref{eq:LV2}, which starts at $\bs{x}(0)=\bs{x}_0$, satisfies$$
\bs{x}(t)\xrightarrow[t\to\infty]{} \bs{x}^*\, .
$$
In the sequel, the term ``stability" will refer to ``asymptotic stability".

\paragraph*{Outline of the article.}
In Section \ref{sec:exist}, we describe sufficient conditions for the existence and uniqueness of a stable equilibrium in the model \eqref{eq:LV2}, see Theorem \ref{th:block_unicite_centered}. Section \ref{sec:surviving_species} is devoted to the study of the properties of the species that survive in each of the communities described by two heuristics. Heuristics \ref{heur:properties_surviving_species} specify the properties of the surviving species and Heuristics \ref{heur:block_distribution} define the distribution of the surviving species.   Finally, in Section \ref{sec:fea} we provide conditions under which the equilibrium is feasible, 
see Theorem \ref{th:fea_2b}.

\section{Existence of a unique equilibrium}\label{sec:exist}

In Figures \ref{fig:scenario1} and \ref{fig:scenario2}, we notice that for different interaction coefficients $\bs{s}$, the system converges to an equilibrium (with or without vanishing species). Theorem \ref{th:block_unicite_centered} below will provide the adequate theoretical framework.

\subsection{Theoretical background}

\paragraph*{Non-invadability condition.}
The research of equilibrium points of \eqref{eq:LV2} is equivalent to the identification of solutions of system \eqref{eq:equilibrium}. However, the number of potential solutions can be extremely large. In order for the equilibrium $\bs{x}^*$ to be stable, there exists a necessary condition, known in ecology as the non-invadability condition \cite{law_permanence_1996}, namely that 
\begin{equation}
\label{eq:block_non_inva_1}
1-x^*_k+(B  \bs{x}^*)_k\le 0\,,\quad \forall k\in [n] .
\end{equation}
In model \eqref{eq:LV2}, the non-invadability condition for a given species whose values at equilibrium is $x_k^*=0$ is equivalent to
\begin{equation}
    \label{eq:block_non_inva_2}
    \left( \frac{1}{x_k} \frac{dx_k}{dt}\right)_{x_k \rightarrow 0^+}\le 0\, .
\end{equation}
Condition \eqref{eq:block_non_inva_2} describes the fact that if we add a species to the system at a very low abundance, it will not be able to invade the system.
As a consequence, the number of possible solutions should solve the following set of constraints:
\begin{equation}\label{eq:block_equilibrium-NI}
\left\{
\begin{array}{lccl}
 x^*_k\left( 1 -  x^*_k + (B\boldsymbol{x^*})_k \right) &=&0&\textrm{for}\ k\in [n]\, ,\\ 
  1 -  x^*_k + (B\boldsymbol{x^*})_k  &\le& 0&\textrm{for}\ k\in [n]\, ,\\
 \bs{x}^* &\geq& 0 & \textrm{componentwise}\, .\\
 \end{array}
\right.
\end{equation}

This casts the search of a nonnegative equilibrium problem into the class of linear complementarity problems (LCP). For a reminder of the definition of an LCP problem, see for instance \cite{clenet_equilibrium_2023}. In the following, we recall the main Theorem for proving the existence and uniqueness of a single equilibrium.

\paragraph{The equilibrium $\bs{x}^*$ and its stability.}

Let $X^\top$ be the transpose of the matrix $X$.

\begin{defi}[Lyapunov diagonal stability]
\label{defi:diag-stable}
A matrix $M$ is called Lyapunov diagonally stable, denoted by $M \in \mathcal{D}$, if and only if there exists a diagonal matrix $D$ with positive diagonal elements such that $DM+M^{\top}D$ is negative definite, i.e. all eigenvalues are negative.
\end{defi}

This class of matrix was already mentioned in Volterra's historical paper \cite{volterra_cons_1931} and in Logofet's book \cite[Chap. 4]{logofet_matrices_1993}, in relation with the stability of LV models.

Recall System \eqref{eq:LV} with different growth rates for each species and consider matrix $B$ is arbitrary,
\begin{equation}\label{eq:generic-LV}
\frac{d\, y_k}{dt} =y_k(r_k +((- I+B)\bs{y})_k)\, ,\quad k\in [n]\, \, .
\end{equation}
The LCP associated with \eqref{eq:generic-LV} is as follows
\begin{equation}\label{eq:lcp_takeuchi}
\left\{
\begin{array}{lccl}
 y^*_k\left( r_k -  y^*_k + (B\boldsymbol{y^*})_k \right) &=&0&\textrm{for}\ k\in [n]\, ,\\ 
  r_k -   y^*_k + (B\boldsymbol{y^*})_k  &\le& 0&\textrm{for}\ k\in [n]\, ,\\
 \bs{y}^* &\geq& 0 & \textrm{componentwise}\, .\\
 \end{array}
\right.
\end{equation}

Takeuchi and Adachi \cite[Th. 3.2.1]{takeuchi_global_1996} establish a criterion for the existence of a unique globally stable equilibrium $\bs{y}^*$ of \eqref{eq:generic-LV}. 

\begin{theo}[Takeuchi and Adachi \cite{takeuchi_existence_1980}] \label{th:block_takeuchi} If $- I+B \in \mathcal{D}$ (see definition \ref{defi:diag-stable}), then System \eqref{eq:lcp_takeuchi} admits a unique solution. In particular, for every $\bs{r}\in \mathbb{R}^n$, there is a unique equilibrium $\bs{y}^*$ to \eqref{eq:generic-LV}, which is globally stable in the sense that for every $\bs{y}_0>0$, the solution to \eqref{eq:generic-LV} which starts at $\bs{y}(0)=\bs{y}_0$ satisfies:
$
\bs{y}(t)\xrightarrow[t\to\infty]{} \bs{y}^*
$.
\end{theo}

\subsection{Sufficient condition for the existence of an equilibrium in model \eqref{eq:LV2}}

As we rely on Theorem \ref{th:block_takeuchi} to assert the existence of a unique equilibrium, we need to understand the asymptotic behaviour of the largest eigenvalue of matrix $H=B+B^\top$. In this quest, the Stieltjes transform $g_\mu$ of a probability measure $\mu$ on the real line 
$$
g_{\mu}(z) =\int \frac{d\mu(\lambda)}{\lambda-z}\,,\quad z\in \mathbb{C}^+=\{ z\in \mathbb{C},\ \textrm{im}(z)>0\}
$$
is a well-known device in Random Matrix Theory since Marchenko and Pastur's groundbreaking paper \cite{marcenko_distribution_1967}. A key feature of the Stieltjes transform is that one can reconstruct probability measure 
$\mu$ knowing $g_{\mu}$, see Proposition \ref{prop:block_stielt_inv}.
We recall some of the properties of the Stieltjes transform in Appendix \ref{app:stieljes} (for more details, see \cite{bai_spectral_2010}). 

In the proof of Theorem \ref{th:block_unicite_centered}, we shall use the fact, established in \cite{ajanki_quadratic_2019}, that the empirical measure of the eigenvalues of matrix $H$:
$$
\frac 1n \sum_{i=1}^n \delta_{\lambda_i(H)}
$$
is well approximated by a fully deterministic distribution whose Stieltjes transform 
$
m(z) =\frac 1n \sum_{i=1}^n m_i(z)$
is defined via the vector $\bs{m}(z) = (m_i(z))$ which is the unique solution of the equation 
$$
-\frac{1}{\bs{m}(z)} = z+S\bs{m}(z)\, ,
$$
carefully defined in the proof of Theorem \ref{th:block_unicite_centered} hereafter, see Eq. \eqref{eq:QVE}.

For a wide range of parameters $(\bs{\beta},\bs{s})$ associated to matrix model $B$, we aim to ensure the existence of a globally stable equilibrium $\bs{x}^*$ of \eqref{eq:LV2} associated to LCP \eqref{eq:block_equilibrium-NI}. Denote by $\left \| \bs{x} \right \|_{\infty}$ the sup norm of a vector and by $\left \| X \right \|_{\infty}$ its induced operator norm, i.e.
$$
\left \| \bs{x} \right \|_{\infty} = \underset{k \in [n]}{\max} \, |x_k|\qquad \textrm{and}\qquad \left \| X \right \|_{\infty} = \underset{k \in [n]}{\max} \sum_{\ell =1}^n |X_{k\ell}|\, .
$$ 
Let $X,Y$ be matrices of the same size, then $X\circ Y$ is their Hadamard product i.e. $(X \circ Y)_{ij} = X_{ij}Y_{ij}$ and consider
$$
\bs{s}\circ \bs{s} = \begin{pmatrix}
s_{11}^2 & s_{12}^2 \\ 
s_{21}^2 & s_{22}^2
\end{pmatrix} \, .
$$

\begin{theo}\label{th:block_unicite_centered} Recall the definition of $B$ in \eqref{eq:mat_int} and $\bs{\beta}$ in \eqref{def:beta} and assume that \begin{equation*}
\left \|\mathrm{diag}(\bs{\beta}) \left((\bs{s} \circ \bs{s})+(\bs{s} \circ \bs{s})^{\top}\right) \right \|_\infty < 1 \, ,
\end{equation*}then a.s. matrix 
$
(I-B)+(I-B)^{\top}
$
is eventually positive definite: with probability one, there exists $N$ depending on matrix $B$'s realization such that for $n\ge N$, $(I-B)+(I-B)^\top$ is positive definite. In particular, $-I+B\in {\mathcal D}$. There exists a unique vector solution to the LCP \eqref{eq:block_equilibrium-NI}. This vector $\bs{x}^*$ is the unique (random) globally stable equilibrium of \eqref{eq:LV2}.
\end{theo}
\begin{rem}
\label{rem:relaxed_assumptions_part1}
Theorem \ref{th:block_unicite_centered} can be extended in two directions. The Gaussianity assumption can be relaxed to any reasonable distribution with finite second moment, and growth rates different from one i.e. $r_k \neq 1$ can be considered (see for instance \cite[Sections 4.2 and 4.3]{bizeul_positive_2021} in the context of a single community).
\end{rem}

\begin{proof}[Sketch of proof] From Theorem \ref{th:block_takeuchi}, we need to verify the Lyapunov diagonally stable condition of the matrix $(-I+B)$ by analyzing its largest eigenvalue
\begin{equation*}
    (-I+B)+(-I+B^{\top}) = -2I+\frac{1}{\sqrt{n}}\begin{pmatrix}
s_{11}(A_{11}+A_{11}^{\top}) & s_{12}A_{12}+s_{21}A_{21}^{\top} \\ 
s_{21}A_{21}+s_{12}A_{12}^{\top} & s_{22}(A_{22}+A_{22}^{\top})
\end{pmatrix}\, .
\end{equation*}
Denote by $H$ the symmetric matrix
\begin{equation*}
    H = \frac{1}{\sqrt{n}}\begin{pmatrix}
H_{11} & H_{12} \\ 
H_{21} & H_{22}
\end{pmatrix} = \frac{1}{\sqrt{n}}\begin{pmatrix}
s_{11}(A_{11}+A_{11}^{\top}) & s_{12}A_{12}+s_{21}A_{21}^{\top} \\ 
s_{21}A_{21}+s_{12}A_{12}^{\top} & s_{22}(A_{22}+A_{22}^{\top})
\end{pmatrix} \, ,
\end{equation*}
where  $H_{ij}$ is a matrix of size $|\mathcal{I}_i|\times|\mathcal{I}_j|$ and each off-diagonal entries follow a Gaussian distribution $\mathcal{N}\left(0,s_{ij}^2+s_{ji}^2\right)$ for all $i,j \in \{1,2\}.$

A matrix is negative definite if and only if all its eigenvalues are negative. Note here that $-2I+H$ is negative definite if and only if the upper eigenvalue of $H$ is less than $2$. The goal of the proof is to give a condition on the parameter $\bs{s}$ such that 
$$
\lambda_{\max} \left ( H \right ) < 2 \, .
$$

Matrix $H$ has a variance profile and such a model has been studied in great details by Erdös \textit{et al.} in \cite{ajanki_universality_2017, ajanki_quadratic_2019}. In particular many properties of its spectrum are associated to properties of deterministic equations named Quadratic Vector Equations (QVE) involving Stieltjes transforms associated to the empirical distribution of the eigenvalues of $H$. These equations have $n$ unknown quantities depending on $z\in \mathbb{C}^+$, stacked into a vector 
$$
\bs{m}(z) = (m_1(z),\cdots,m_n(z))\ ,
$$
and related by the following system of equations (QVE):
\begin{equation*}
\begin{split}
    k \in \mathcal{I}_1 \quad,\quad -\frac{1}{m_k(z)} &= z + \sum_{\ell \in \mathcal{I}_1}\frac{2s_{11}^2}{n} m_\ell(z) + \sum_{\ell \in\mathcal{I}_2}\frac{1}{n}\left(s_{12}^2+s_{21}^2\right)  m_\ell(z)\, , \\
    k \in \mathcal{I}_2 \quad ,\quad  -\frac{1}{m_k(z)} &= z + \sum_{\ell \in \mathcal{I}_1}\frac{1}{n}\left(s_{12}^2+s_{21}^2\right)  m_\ell(z) + \sum_{\ell \in \mathcal{I}_2}\frac{2s_{22}^2}{n} m_\ell(z)\,.
\end{split}
\end{equation*}
Denote by $1/\bs{m}(z) = (1/m_1(z),\cdots,1/m_n(z))^\top$, $\bs{1}_{\mathcal{I}_i}$ a vector whose entries are 1's of size $|\mathcal{I}_i|$ and
$$S = \frac{1}{n}\begin{pmatrix}
2 s_{11}^2\bs{1}_{\mathcal{I}_1}\bs{1}_{\mathcal{I}_1}^{\top} & (s_{12}^2+s_{21}^2)\bs{1}_{\mathcal{I}_1}\bs{1}_{\mathcal{I}_2}^{\top} \\ 
(s_{12}^2+s_{21}^2)\bs{1}_{\mathcal{I}_2}\bs{1}_{\mathcal{I}_1}^{\top} & 2s_{22}^2\bs{1}_{\mathcal{I}_2}\bs{1}_{\mathcal{I}_2}^{\top} 
\end{pmatrix}\, ,$$
the QVE equations can be written in the more compact form
\begin{equation}
    -\frac{1}{\bs{m}(z)} = z+S\bs{m}(z)\, .
    \label{eq:QVE}
\end{equation}
Following Theorem 2.1 in Ajanki \textit{et al.} \cite{ajanki_quadratic_2019}, for all $z \in \mathbb{C}^+$, Equation \eqref{eq:QVE} has a unique solution $\bs{m} = \bs{m}(z)$ and 
$$
\frac 1n \sum_{i=1}^n m_i(z)
$$ is the Stieltjes transform of a probability measure, the support of which is included in $[-\Sigma,\Sigma]$, where $\Sigma = 2\left \| S \right \|_\infty^{1/2}$. 

This information gives an asymptotic bound on the support of the matrix $H$ associated with \eqref{eq:QVE}, i.e. asymptotically $\forall \, \varepsilon >0$ there exists $N$ depending on matrix $B$'s realization, such that for $n \geq N$
$$
\lambda_{\max}\left( H \right ) \leq 2\left \| S \right \|_\infty^{1/2}+\varepsilon\, .
$$ 
Recall that $-2I+H$ is negative definite iff $\lambda_{\max}\left( H \right ) < 2$. This condition is fulfilled if 
\begin{equation*}
2\left \| S \right \|_\infty^{1/2} < 2 \, ,
\end{equation*}
or equivalently
$$
\left \| S \right \|_\infty < 1 \, . 
$$
Note that this condition is sufficient but not necessary. Given the particular shape of the matrix $S$, computing its norm is equivalent to computing the norm of a matrix of size $2\times 2$
$$
\left \| S \right \|_\infty = \left \|\mathrm{diag}(\bs{\beta}) \left((\bs{s}\circ\bs{s})+(\bs{s}\circ\bs{s})^{\top}\right) \right \|_\infty \, 
= \left \| \begin{pmatrix}
\beta_1 & 0\\ 
0 & \beta_2
\end{pmatrix}\begin{pmatrix}
2s_{11}^2 & s_{12}^2+s^2_{21}\\
s^2_{12}+s^2_{21} & 2s_{22}^2
\end{pmatrix} \right \|_\infty\,, 
$$
which completes the proof. We can then rely on Theorem \ref{th:block_takeuchi} to conclude.
\end{proof}

\begin{rem}
~
\begin{enumerate}
    \item In the context of a unique community, suppose that $\bs{s} = s\bs{1}\bs{1}^{\top}$, then the previous condition takes the simpler form $s < 1/\sqrt{2}$, a condition already mentionned in \cite{clenet_equilibrium_2023}. Indeed, starting from the condition of Theorem \ref{th:block_unicite_centered}, the condition on the matrix is
    $$
    \left \|\begin{pmatrix}
1/2 & 0\\ 
0 & 1/2
\end{pmatrix}^{1/2} \begin{pmatrix}
2s^2 & 2s^2\\ 
2s^2 & 2s^2
\end{pmatrix}\begin{pmatrix}
1/2 & 0\\ 
0 & 1/2
\end{pmatrix}^{1/2} \right \|_\infty = \left \| \begin{pmatrix}
1s^2 & 1s^2\\ 
1s^2 & 1s^2
\end{pmatrix} \right \|_\infty = 2s^2 \, ,
$$
The same sufficient condition is obtained $2s^2 < 1 \ \Leftrightarrow \ s <1/\sqrt{2}$.
\item The condition given in Theorem \ref{th:block_unicite_centered} is sufficient to guarantee a stable unique solution to LCP \eqref{eq:block_equilibrium-NI} but not necessary. Even in the single community case, finding optimal thresholds remains an open question (see \cite[Table 1]{akjouj_complex_2022}).
\end{enumerate}
\end{rem}

\section{Surviving species}
\label{sec:surviving_species}
In Section \ref{sec:exist}, we have given conditions on matrix $\bs{s}$ and on $\bs{\beta}=(\beta_1,\beta_2)$ for the existence of a globally stable equilibrium
$\bs{x}^*$ to \eqref{eq:LV2} under the non-invadability condition. The equilibrium vector $\bs{x}^*$ is random and depends on the realization of matrix $B$. Moreover since $\bs{s}$ has fixed components and does not depend on $n$, the equilibrium $\bs{x}^*$ will feature vanishing components (see the original argument for a unique community in \cite{dougoud_feasibility_2018} and the discussion in \cite{bizeul_positive_2021}).
In an ecological context, we differentiate two kind of components in vector $\bs{x^*}$, the non-vanishing components $x_k^*>0$ corresponding to surviving species and the vanishing ones $x_k^*=0$ corresponding to the species going to extinction:
$$
x_k(t)\xrightarrow[t\to\infty]{} 0\, .
$$ 
Hereafter, we describe heuristics of the statistical properties of $\bs{x}^*$: the proportion of surviving species in each community, the distribution of the corresponding abundances, which turns out to be a truncated Gaussian, etc.

\subsection{Heuristics for the properties of surviving species}
Starting from the model \eqref{eq:LV2}, the set of surviving species in community $i \in \{1,2\}$ is defined as

\begin{eqnarray*}
    \mathcal{S}_1 &=& \{k \in \mathcal{I}_1, \ x_k^* > 0 \}\ ; \quad \mathcal{I}_1 = [1,\beta_1n] \, ,\\
    \mathcal{S}_2 &=& \{k \in \mathcal{I}_2, \ x_k^* > 0 \} \ ; \quad \mathcal{I}_2 = [\beta_1n+1,n] \, .
\end{eqnarray*}
Given the random equilibrium $\bs{x}^*$, we introduce the following quantities for each community $i\in \{1,2\}$
\begin{equation*}
    \hat{p}_i = \frac{\left | \mathcal{S}_i \right |}{\left |\mathcal{I}_i\right |} \quad,\quad  \hat{m}_i = \frac{1}{\left | \mathcal{S}_i \right |}\sum_{k \in \mathcal{I}_i} x_k^*\quad,\quad \hat{\sigma}_i^2=\frac{1}{\left | \mathcal{S}_i \right |}\sum_{k \in \mathcal{I}_i} (x_k^*)^2 \, .
\end{equation*}
Quantity $\hat{p}_i$ represents the proportion of surviving species in community $i$,  $\hat{m}_i$ the empirical mean of the abundances of the surviving species in community $i$ and $\hat{\sigma}_i^2$, the empirical mean square of the surviving species in community $i$. 

Denote by $Z\sim{\mathcal N}(0,1)$ a standard Gaussian random variable and by $\Phi$ the cumulative Gaussian distribution function:
$$
\Phi(x)=\int_{-\infty}^x \frac{e^{-\frac{u^2}2}}{\sqrt{2\pi}}\, du\, .
$$

In Heuristic \ref{heur:properties_surviving_species}, we derive the properties of surviving species, specifically the proportion of surviving species and the mean square of surviving species. The presentation of these properties is illustrated in Fig. \ref{fig:block_test} through the use of numerical simulations to support this heuristic. In addition, the technical details of how to obtain the heuristic are presented in the subsequent paragraphs.

\begin{heur}
\label{heur:properties_surviving_species}
Let $\bs{s}$ be the $2 \times 2$ matrix of interaction strengths and assume that the condition of Theorem \ref{th:block_unicite_centered} holds, then the following system of four equations and four unknowns $(p_1,p_2,\sigma_1,\sigma_2)$
\begin{eqnarray}
\label{eq:heur1}
p_1 &=&1- \Phi(\delta_1) \ , \\
\label{eq:heur2}
p_2 &=& 1- \Phi(\delta_2) \ , \\
\label{eq:heur3}
(\sigma_1)^2 &=&1+2\Delta_1 \mathbb{E}(Z|Z>\delta_1)+\Delta_1^2\mathbb{E}(Z^2|Z>\delta_1) \ , \\
\label{eq:heur4}
(\sigma_2)^2 &=& 1+2\Delta_2\mathbb{E}(Z|Z>\delta_2)+\Delta_2^2\mathbb{E}(Z^2|Z>\delta_2) \ , 
\end{eqnarray}
where
\begin{equation}
\label{eq:def_var_heur}
\Delta_i = \sqrt{p_1 (\sigma_1)^2\beta_1s_{i1}^2+p_2 (\sigma_2)^2\beta_2s_{i2}^2} \quad \textrm{and}\quad  \delta_i = \frac{-1}{\Delta_i}\,, 
\end{equation}
admits a unique solution $(p_1^*,p_2^*,\sigma_1^*,\sigma_2^*)$ and for $i \in \{1,2\}$
\begin{equation}\label{eq:convergence-heuristics1}
 \hat{p}_i \xrightarrow[n\to\infty]{a.s.} p_i^* \qquad \text{and}\qquad \hat{\sigma}_i\xrightarrow[n\to\infty]{a.s.} \sigma^*_i\, .
\end{equation}
In the sequel, denote by $
    \Delta^*_i = \sqrt{p_1^* (\sigma_1^*)^2\beta_1s_{i1}^2+p^*_2 (\sigma_2^*)^2\beta_2s_{i2}^2}$ and 
    $\delta^*_i = -1/{\Delta^*_i}$.
\end{heur}

\begin{rem}
    In Section \ref{subsec:heuristics-building}, we explain how to establish Eq. \eqref{eq:heur1}-\eqref{eq:heur4}. Both convergences in \eqref{eq:convergence-heuristics1} are part of the statement of the Heuristics and are justified by the fact that the starting points to establish Eq. \eqref{eq:heur1}-\eqref{eq:heur4} are the empirical quantities $\hat{p}_i$ and $\hat{\sigma}_i$.
\end{rem}

There is a strong matching between the solutions obtained by solving \eqref{eq:heur1}-\eqref{eq:heur4} and their empirical counterparts obtained by Monte-Carlo simulations. This is illustrated in Fig. \ref{fig:block_test}.
\begin{figure}[ht]
\centering
\begin{subfigure}{0.48\textwidth}
  \includegraphics[width=\textwidth]{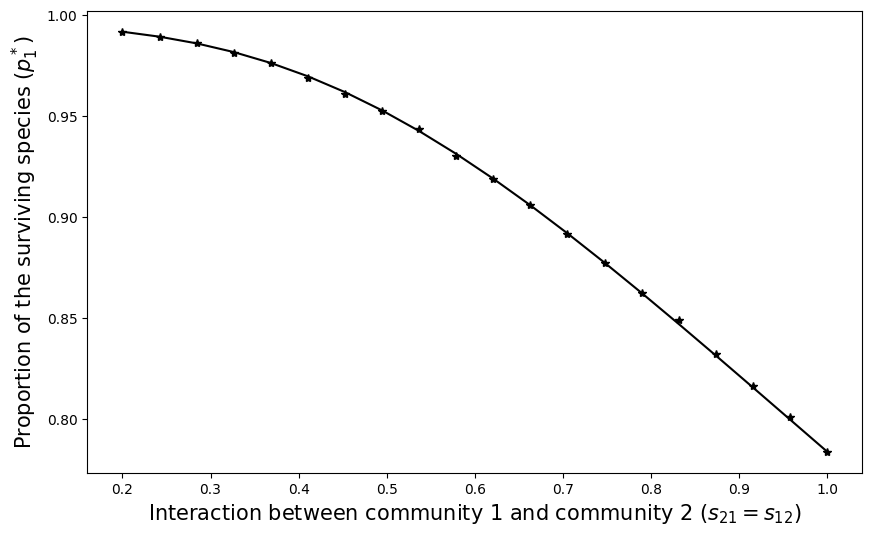}
\end{subfigure}
\hfill
\begin{subfigure}{0.48\textwidth}
  \includegraphics[width=\textwidth]{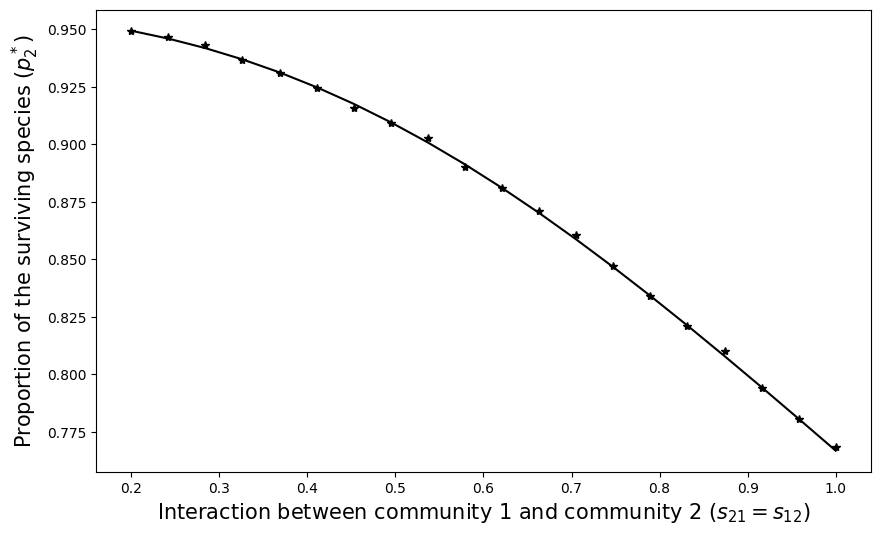}
  
\end{subfigure}
  
  
\hfill
\begin{subfigure}{0.48\textwidth}
  \includegraphics[width=\textwidth]{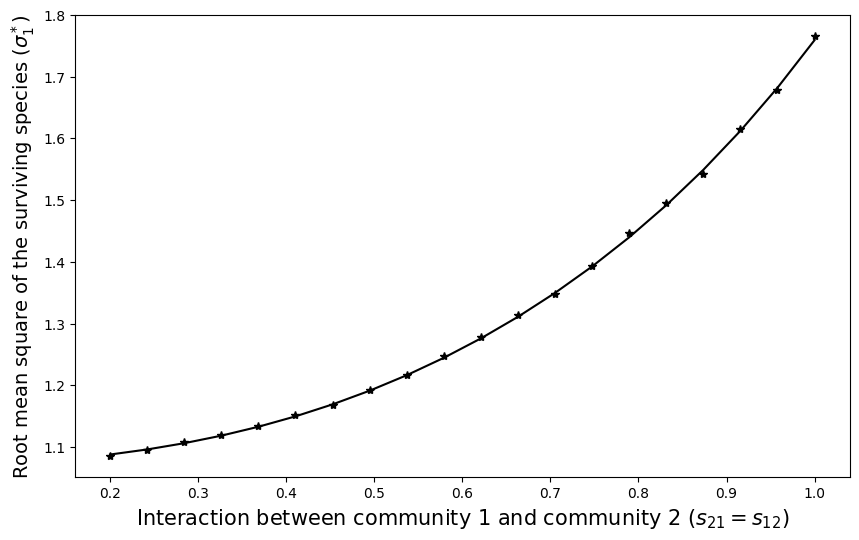}
  \caption{Parameters $(p_1^*,\sigma_1^*)$ versus $(s_{21},s_{12})$.}

\end{subfigure}
\hfill
\begin{subfigure}{0.48\textwidth}
  \includegraphics[width=\textwidth]{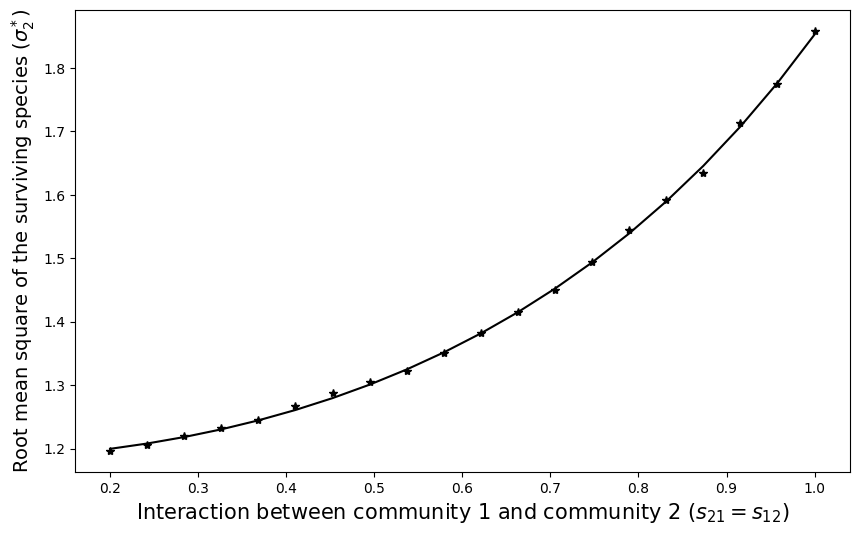}
  \caption{Parameters $(p_2^*,\sigma_2^*)$ versus $(s_{21},s_{12})$.}

\end{subfigure}
\caption{Comparison between the theoretical solutions $(p_1^*,p_2^*,\sigma_1^*,\sigma_2^*)$ of \eqref{eq:heur1}-\eqref{eq:heur4} and their empirical Monte Carlo counterpart (the star marker) as functions of the off-diagonal block interaction strength $(s_{12},s_{21})$. The left column is associated to the properties of community 1. The right column is associated to the properties of community 2. Matrix B has size $n=500$ and the number of Monte Carlo experiments is 500. The parameters are
$
\bs{s} = \begin{pmatrix}
1/3 & s_{12}\\ 
s_{21} & 1/\sqrt{2}
\end{pmatrix} \ , \ \bs{\beta} = \left(\frac{1}{2},\frac{1}{2} \right) \, .
$
When off-diagonal block interactions $s_{12},s_{21}$ increase, the proportion of surviving species $p^*$ decreases but the root mean square of their equilibrium abundances $\sigma^*$ increases. For any numerical details, please refer to Appendix \ref{app:num}.}
\label{fig:block_test}
\end{figure}

\subsection{Construction of the heuristics}\label{subsec:heuristics-building}
Obtaining information about the fixed point is equivalent to solving the LCP problem
\begin{equation*}
    x_k^*\left(1-x_k^*+\sum_{\ell=1}^n B_{k\ell}x_\ell^*\right) = 0 \ , \ \forall k \in [n] \,.
\end{equation*}
Consider the random variables:
\begin{equation*}
    \forall k \in [n], \ \check{Z}_k = \sum_{\ell \in \mathcal{S}_1 \cup \mathcal{S}_2}B_{k\ell}x_\ell^* \, .
\end{equation*}
We assume that asymptotically the $x_\ell^*$'s are independent from the $B_{k\ell}$'s, an assumption supported by the chaos hypothesis, see for instance Geman and Hwang \cite{geman_chaos_1982}. 

Denote by $\mathbb{E}_{\bs{x}^*}=\mathbb{E}(\,\cdot\mid \bs{x}^*)$ the conditional expectation with respect to $\bs{x}^*$. Notice that conditionally to $\bs{x}^*$, the $\check{Z}_k$'s are independent Gaussian random variables, whose first two moments can easily be computed, see Appendix \ref{app:proof_heuristics} for the details:
\begin{equation*}
    \forall k \in \mathcal{I}_i, \quad \textrm{Var}_{\bs{x}^*}(\check{Z}_k) \simeq \hat{p}_1 \hat{\sigma}_1^2\beta_1s_{i1}^2+\hat{p}_2 \hat{\sigma}_2^2\beta_2s_{i2}^2 \, .
\end{equation*}
Notice that the fact that $\textrm{Var}_{\bs{x}^*}(\check{Z}_k)$ only depends on $\hat{p}_1,\,\hat{p}_2, \, \hat{\sigma}_1, \, \hat{\sigma}_2$. It is assumed that the estimators, namely, $\hat{p}_1, \hat{p}_2, \hat{\sigma}_1, \hat{\sigma}_2$, are converging quantities when $n\rightarrow\infty$ to their respective limits, $p_1^*, p_2^*, \sigma_1^*, \sigma_2^*$. This assumption is consistent with the chaos hypothesis and is fundamental to the derivation of the heuristics.

Recall that 
\begin{equation*}
    \Delta^*_i = \sqrt{p_1^* (\sigma_1^*)^2\beta_1s_{i1}^2+p^*_2 (\sigma_2^*)^2\beta_2s_{i2}^2}\quad \text{and}\quad \delta^*_i = \frac{-1}{\Delta^*_i}\, . 
\end{equation*}
Then the convergence of $\hat{p}_1,\,\hat{p}_2, \, \hat{\sigma}_1, \, \hat{\sigma}_2$ supports the idea that $\check{Z}_k$ is unconditionally a Gaussian random variable with second moment:
$$
\textrm{Var}(\check{Z}_k)= (\Delta^*_i)^2 \, ,
$$
where $\Delta^*_i$ corresponds to the average variance of the interactions on community $i$ which depends on four parameters $p_1^*,\, p_2^*, \sigma_1^*, \, \sigma_2^*$. We can introduce two families of standard Gaussian random variables $(Z_k)_{k\in \mathcal{I}_1}$ and $(Z_k)_{k\in \mathcal{I}_2}$:
$$
\forall k \in \mathcal{I}_i, \ \  Z_k = \frac{\check{Z}_k}{\sqrt{\textrm{Var}(\check{Z}_k)}}= \frac{\check{Z}_k}{\Delta^*_i}.
$$    

Consider the equilibrium $\bs{x}^* = (x_k^*)_{k\in[n]}$, the definition of the LCP equilibrium implies if $k \in \mathcal{S}_1 \cup \mathcal{S}_2$:
$$
x_k^*(1-x_k^*+(Bx^*)_k)=0 \quad \textrm{and}\quad  1+(B\bs{x}^*)_k = 1+\check{Z}_k > 0
\, .$$
We finally obtain the following relationship for the surviving species:
\begin{equation}
\label{eq:relation_pers_species}
    x_k^* = 1+\Delta^*_i Z_k \quad \text{for} \quad k\in \mathcal{S}_i \, . 
\end{equation}
Note that $\Delta^*_i$ corresponds to the average variance of the interactions on community $i$.

\paragraph*{Heuristics \eqref{eq:heur1}-\eqref{eq:heur2}.}
We can write the first two equations:
\begin{eqnarray*}
    \mathbb{P}(x_k^* > 0 \mid k \in \mathcal{I}_1) &=& \mathbb{P}(Z_k > \delta^*_1 \mid k \in \mathcal{I}_1) =  1- \Phi(\delta^*_1)\, ,\\
    \mathbb{P}(x_k^* > 0 \mid k \in \mathcal{I}_2) &=& \mathbb{P}(Z_k > \delta^*_2 \mid k \in \mathcal{I}_2) = 1- \Phi(\delta^*_2) \, ,
\end{eqnarray*}
and finally obtain \eqref{eq:heur1} and \eqref{eq:heur2}:
\begin{equation*}
p^*_1 =1- \Phi(\delta^*_1) \quad \textrm{and}\quad p^*_2 =1- \Phi(\delta_2^*) \, .
\end{equation*}

\paragraph*{Heuristics \eqref{eq:heur3}-\eqref{eq:heur4}.}
Our starting point is the following generic representation of an abundance at equilibrium (either of a surviving or vanishing species) in the case $k \in \mathcal{S}_i$:
\begin{equation*}
x_k^* = \left( 1 +\Delta^*_i Z_k \right)\bs{1}_{\{Z_k>  \delta^*_i\}} = \bs{1}_{\{Z_k>  \delta^*_i\}} +\left( \Delta^*_i Z_k \right)\bs{1}_{\{Z_k>  \delta^*_i\}} \, .
\end{equation*}

Taking the square, we get:
\begin{multline*}
(x_k^*)^2 = \left( 1+\Delta^*_iZ_k \right)^2\bs{1}_{\{Z_k>  \delta^*_i\}} = \bs{1}_{\{Z_k>  \delta^*_i\}} + 2\Delta^*_i Z_k \bs{1}_{\{Z_k>  \delta^*_i\}} +\left( (\Delta^*_i)^2 Z^2_k \right)\bs{1}_{\{Z_k>  \delta^*_i\}} \, .
\end{multline*}
Summing over $\mathcal{S}_i$ and normalizing, we get
\begin{equation*}
\begin{split}
\frac{1}{|{\mathcal S}_i|}\sum_{k \in \mathcal{S}_i}(x_k^*)^{2} &= \frac{1}{|{\mathcal S}_i|}\sum_{k \in \mathcal{S}_i}\bs{1}_{\{Z_k > -\delta_i^*\}}
    +2\Delta^*_i\frac{1}{|{\mathcal S}_i|}\sum_{k \in \mathcal{S}_i}Z_k\bs{1}_{\{Z_k > -\delta_i^*\}}+(\Delta^*_i)^2\frac{1}{|{\mathcal S}_i|}\sum_{k \in \mathcal{S}_i}Z^2_k\bs{1}_{\{Z_k > \delta_i^*\}} \,, \\
\hat{\sigma}_i^2 &\stackrel{(a)}= 1+2\Delta^*_i\frac{|\mathcal{I}_i|}{|{\mathcal S}_i|}\frac{1}{|\mathcal{I}_i|}\sum_{k \in \mathcal{I}_i}Z_k\bs{1}_{\{Z_k > \delta^*_i\}}+(\Delta^*_i)^2\frac{|\mathcal{I}_i|}{|{\mathcal S}_i|}\frac{1}{|\mathcal{I}_i|}\sum_{k \in \mathcal{I}_i}Z^2_k\bs{1}_{\{Z_k > \delta_i^*\}} , \\
\hat{\sigma}_i^2 &\stackrel{(b)}\simeq 1+2\Delta^*_i \frac{1}{\mathbb{P}(Z > \delta^*_i)}\mathbb{E}(Z\bs{1}_{\{Z > \delta_i^*\}})+(\Delta^*_i)^2 \frac{1}{\mathbb{P}(Z > \delta^*_i)}\mathbb{E}(Z^2\bs{1}_{\{Z > \delta_i^*\}}), \\
\hat{\sigma}_i^2 &\simeq 1+2\Delta^*_i\mathbb{E}(Z \mid Z > \delta^*_i)+(\Delta^*_i)^2\mathbb{E}(Z^2 \mid Z > \delta^*_i).
\end{split}
\end{equation*}
where $(a)$ follows from the fact that $|{\mathcal S}_i| = \sum_{k \in \mathcal{S}_i}\bs{1}_{\{Z_k > \delta^*_i \}}$ (by definition of ${\mathcal S}_i$), $(b)$ from the law of large numbers $\frac{1}{|\mathcal{I}_i|} \sum_{k\in \mathcal{I}_i} Z^j_k \bs{1}_{\{Z_k>\delta_i\}} \xrightarrow[n\to\infty]{} \mathbb{E}Z^j \bs{1}_{\{Z>\delta^*_i\}},\quad j=1,2$ and $\frac{|{\mathcal S}_i|}{|\mathcal{I}_i|} \xrightarrow[n\to\infty]{} \mathbb{P}(Z>\delta^*_i)$ with $Z\sim{\mathcal N}(0,1)$. It remains to replace $\hat{\sigma}_i$ by its limit $\sigma_i^*$ to obtain \eqref{eq:heur3}-\eqref{eq:heur4}. 
We finally obtain the third and fourth equations:

\begin{eqnarray*}
    (\sigma^*_1)^2 &=&(1+\lambda^*_1)^2+2(1+\lambda^*_1)\Delta^*_1 \mathbb{E}(Z\mid Z>\delta^*_1)+(\Delta^*_1)^2\mathbb{E}(Z^2\mid Z>\delta^*_1)\ ,\\
    (\sigma^*_2)^2 &=& (1+\lambda^*_2)^2+2(1+\lambda^*_2)\Delta^*_2\mathbb{E}(Z\mid Z>\delta^*_2)+(\Delta^*_2)^2\mathbb{E}(Z^2\mid Z>\delta^*_2)\ .
\end{eqnarray*}

\subsection{General properties of the ecosystem}
The properties at equilibrium, such as the proportion and mean square of the abundance of surviving species, can be computed for each community by solving the system of equations in Heuristics \ref{heur:properties_surviving_species}. An additional property, the mean abundance of the surviving species at equilibrium for each community ($m_1^*, m_2^*$), can be calculated using a method similar to the mean square of the abundances (see Appendix \ref{app:heur_mean} for the details of the computations).
\begin{align}
\label{eq:mean_c1}
m_1^* &= 1+\Delta^*_1 \mathbb{E}(Z|Z>\delta^*_1) \ , \\
\label{eq:mean_c2}
m_2^* &= 1+\Delta^*_2 \mathbb{E}(Z|Z>\delta^*_2) \ .    
\end{align}
The two equations are not necessary for solving Heuristics \ref{heur:properties_surviving_species}, but they provide new information. In particular, strong inter- or intra-community interactions increase the mean abundance of the surviving species (see Fig. \ref{fig:distrib_surv}).

Conditional on each community, one can easily extend the properties of each community to the whole ecosystem. We denote by $p^*$ the proportion, $(\sigma^*)^2$ the mean square and $m^*$ the mean of surviving species. We observe the linear effect of community size $\bs{\beta}$ on general properties:
\begin{enumerate}
    \item Proportion of surviving species.
\begin{align*}
     \mathbb{P}(x_k^* > 0) &= \mathbb{P}(x_k^* > 0 \mid k \in \mathcal{I}_1)\mathbb{P}(k\in \mathcal{I}_1)+\mathbb{P}(x_k^* > 0 \mid k \in \mathcal{I}_2)\mathbb{P}(k\in \mathcal{I}_2) \, , \\
     p^* &= p^*_1 \beta_1+p^*_2 \beta_2 \, .
\end{align*}
\item Mean square of the abundance of the surviving species.
\begin{align*}
\mathbb{E}((x_k^*)^2) &= \mathbb{E}((x_k^*)^2\mid k \in \mathcal{I}_1)\mathbb{P}(k\in \mathcal{I}_1)+\mathbb{E}((x_k^*)^2 \mid k \in \mathcal{I}_2)\mathbb{P}(k\in \mathcal{I}_2) \, , \\
(\sigma^*)^2 &= (\sigma^*_1)^2 \beta_1 + (\sigma^*_2)^2 \beta_2 \, .
\end{align*}

\item Mean of the abundance of the surviving species
\begin{align*}
     \mathbb{E}(x_k^*) &= \mathbb{E}(x_k^{*} \mid k \in \mathcal{I}_1)\mathbb{P}(k\in \mathcal{I}_1)+\mathbb{E}(x_k^{*} \mid k \in \mathcal{I}_2)\mathbb{P}(k\in \mathcal{I}_2) \, , \\
    m^* &= m^*_1 \beta_1 + m^*_2 \beta_2 \, .
\end{align*}
\end{enumerate}

This linear relationship illustrates that the impact of an ecological community on the benefits of the entire ecosystem is directly proportional to its size. In other words, a larger ecological community will have a greater influence on the ecosystem than a smaller ecological community.

\subsection{Distribution of the surviving species}
We may recall the following representation of the abundance $x_k^*$ of a surviving species when $k \in \mathcal{S}_i$:
$$
x_k^* = 1+\Delta^*_i Z_k \text{ if } k\in \mathcal{S}_i,
$$
where $Z_k \sim \mathcal{N}(0,1)$ and $Z_k > \delta_i^* = \delta_i(p_i^*,\sigma_i^*)$ defined in \eqref{eq:relation_pers_species}.
This representation allows to characterize the distribution of $x_k^*$ of each community. It turns out that the surviving species of each community follow a truncated Gaussian distribution.
\begin{heur}
\label{heur:block_distribution}
Let $\bs{s}$ be the $2 \times 2$ matrix of interaction strengths and assume that the condition of Theorem \ref{th:block_unicite_centered} holds and let $(p^*_1,p_2^*,\sigma_1^*,\sigma_2^*)$ be the solution of the system \eqref{eq:heur1}-\eqref{eq:heur4}. Recall the definition \eqref{eq:def_var_heur} of $\Delta_i$ and $\delta_i$ and denote by $\delta_i^* = \delta_i(p_i^*,\sigma_i^*)$ and $\Delta_i^* = \Delta_i(p_i^*,\sigma_i^*)$. Let $x^*_k>0$ be a positive component of $\bs{x}^*$ belonging to the community $i$, then the law of $x^*_k$ is
$$
\mathcal{L}(x_k^*) \xrightarrow[n\to\infty]{} \mathcal{L}\left(1+\Delta^*_i Z \ \mid\ Z > \delta_i^* \right)\ ,
$$
where $Z\sim {\mathcal N}(0,1)$. Otherwise stated, asymptotically for $ k \in \mathcal{S}_i, \, x^*_k$ admits the following density
\begin{equation}
\label{eq:distribution}
 f_k(y) = \frac{\bs{1}_{\{y>0\}}}{\Phi(-\delta^*_i)}\frac{1}{\Delta^*_i\sqrt{2\pi}}\,
 \exp\left\{-\frac 12\left(\frac{y}{\Delta^*_i}+\delta^*_i\right)^2\right\}\ .
\end{equation}
\end{heur}
The heuristics are derived from the fact that, from equation \eqref{eq:relation_pers_species}, if $x_k^*$ is a surviving species and $k \in \mathcal{S}_i$ then
$$
x_k^* = 1+\Delta^*_i Z_k \, ,
$$
conditionally on the fact that the right-hand side of the equation is positive, that is $Z_k>\delta_i^*$. A simple change of variable yields the density - details are provided in Appendix \ref{app:proof_heuristics}.

Fig. \ref{fig:distrib_surv} illustrates the matching between the theoretical distribution obtained by equation \eqref{eq:distribution} and a histogram obtained by generating the interaction matrix for $2$ communities. 
In Fig. \ref{fig:distrib_uniform}, the validity of heuristics in the case of non-Gaussian entries is illustrated.

\begin{rem}
\label{rem:relaxed_assumptions_part2}
The proof relies on the Gaussianity assumption, but we are convinced that it could be extended beyond. In particular, in Figure \ref{fig:distrib_uniform}, non-Gaussian entries centered $\mathbb{E}(B_{k\ell}) = 0$ with variance one $\mathbb{E}(|B_{k\ell}|^2) =1$ are considered. The distribution of surviving species still fits the truncated Gaussian in this case.
\end{rem}

\begin{figure}[ht]
    \centering
    \includegraphics[width=0.8\textwidth]{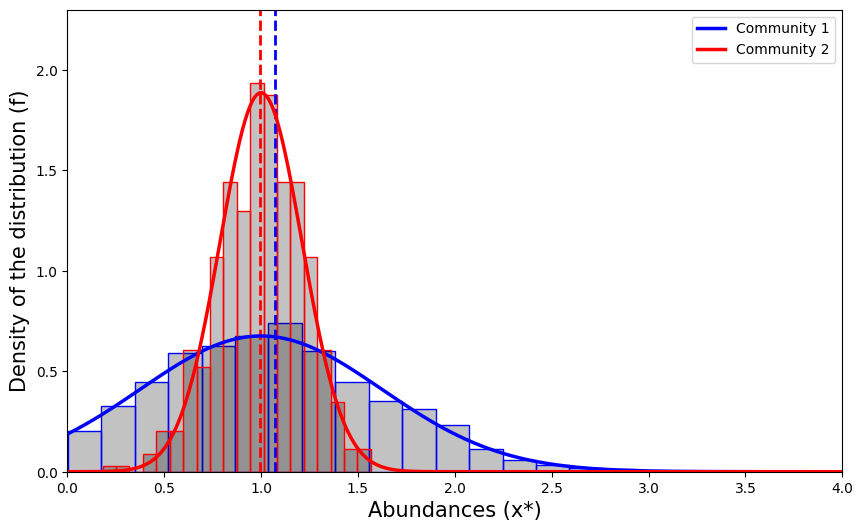}
    \captionsetup{singlelinecheck=off}
    \caption[Distribution of surviving species in each community 1.]{
Distribution of surviving species in each community. The $x$-axis represents the value of the abundances and the histogram is built upon the positive components of equilibrium $\bs{x}^*$ associated to each community. The blue-solid line (resp. red-solid line) represents the theoretical distribution of community 1 (resp. community 2) for parameters $\bs{s}$ as given by Heuristics \ref{heur:block_distribution}. The blue-dashed vertical line (resp. red-dashed vertical line) corresponds to the mean abundance of community 1 (resp. community 2) given by equations \eqref{eq:mean_c1}-\eqref{eq:mean_c2}. All the matrix entries $A_{ij}$'s are independent Gaussian ${\mathcal N}(0,1)$; the parameters are set to $n = 2000$, $\bs{\beta} = \left(0.75,0.25\right)$ and  $\bs{s} = {\scriptsize \begin{pmatrix}
1/2 & 1/\sqrt{2}\\ 
1/5 & 1/9
\end{pmatrix}}$.
}
\label{fig:distrib_surv}
\end{figure}
\begin{figure}[ht]
    \centering
    \includegraphics[width=0.8\textwidth]{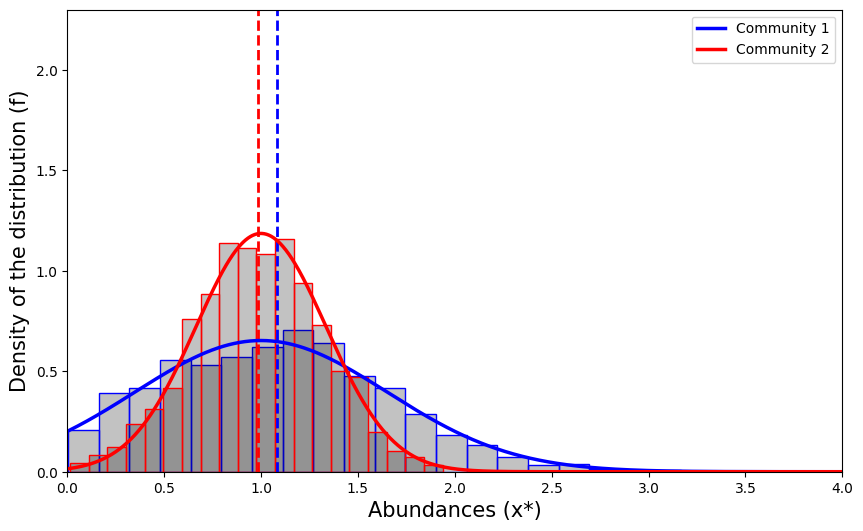}
    \captionsetup{singlelinecheck=off}
    \caption[Distribution of surviving species in each community 2.]{
Distribution of surviving species in each community with non-Gaussian entries. The $x$-axis represents the value of the abundances. The histogram is built upon the positive components of equilibrium $\bs{x}^*$ associated to each community.
(blue for community 1, red for community 2). The solid lines represent the theoretical distributions associated to parameters $\bs{s}$ as given by Heuristics \ref{heur:block_distribution}. The blue-dashed vertical line (resp. red-dashed vertical line) corresponds to the mean abundance of community 1 (resp. community 2). The entries of the $A_{ij}$ matrices are uniform $\mathcal{U}(-\sqrt{3},\sqrt{3})$ with variance $1$ and the parameters are set to $n=2000$, $\bs{\beta}=(0.5, 0.5)$ and $\bs{s} = {\scriptsize \begin{pmatrix}
1/2 & 2/3\\ 
1/3 & 1/4
\end{pmatrix}}$. Notice in particular that the histogram is well predicted by the theoretical distributions even if the entries are non-Gaussian.}
\label{fig:distrib_uniform}
\end{figure}

\section{Feasibility} \label{sec:fea}

Recall $s$ the interaction parameter in the case of a unique community. According to the work of Dougoud \textit{et al.}. \cite{dougoud_feasibility_2018}, if $s$ is fixed (i.e. does not depend on $n$) then there can be no feasible equilibrium at large $n$. 
Following this work, Bizeul and Najim \cite{bizeul_positive_2021} provided the appropriate normalization of $s$ to have a feasible equilibrium. The threshold corresponds to $s \sim 1/\sqrt{2\log(n)}$. The equilibrium is feasible almost surely when $s$ is less than this threshold value, 
i.e. when elements of random matrix $B$ are divided by $\sqrt{2n\log(n)}$ or a larger factor. Some extensions of these results have been made in the sparse case \cite{akjouj_feasibility_2022} and with a mean and pairwise correlated entries \cite{clenet_equilibrium_2022}. In this section, conditions are given on the matrices $\bs{s}$ to get a feasible equilibrium in each community, called co-feasibility. We then provide some ecological interpretations.

\subsection{Theoretical analysis of the threshold}
Recall the notation $\bs{x} = (x_k)_{k \in [n]}$ and denote by $\left \| \bs{x} \right \|_{\infty} = \underset{k \in [n]}{\max} \, |x_k|$.
We are interested in the existence of a feasible solution of the fixed point problem associated with the model \eqref{eq:LV2}. To consider this problem, we extend the computations of Bizeul and Najim in the framework of a block structure network. Consider $\bs{s}$ such that $I-B$ is invertible. The problem is defined by 
\begin{equation}
    \bs{x^*} = \bs{1}+B\bs{x^*} \ \Leftrightarrow \ \bs{x^*} = (I-B)^{-1} \bs{1}\, ,
    \label{eq:fea_fixed}
\end{equation}
The problem \eqref{eq:fea_fixed} admits a unique solution. We consider a matrix $\bs{s}$ which depends on $n$, i.e. $\bs{s}=\bs{s}_n$ such that:
$$
\bs{s}_n \xrightarrow[n\to\infty]{} \bs{0} \quad \Leftrightarrow \quad \forall \, i,j \in \{1,2\}\, ,\, s_{ij} \xrightarrow[n\to\infty]{} 0 \, .
$$
Note that for sufficiently large $n$, the problem satisfies the sufficient condition of Theorem \ref{th:block_unicite_centered} to have a unique globally stable equilibrium, which in this case is a feasible equilibrium.

Let matrix $B_n$ depending on the interaction matrix $\bs{s_n}$ defined by
\begin{equation}
B_n = V\bs{s}_nV^{\top}\circ\frac{1}{\sqrt{n}}\begin{pmatrix}
A_{11} & A_{12} \\ 
A_{21} & A_{22}
\end{pmatrix}\, ,
\label{eq:mat_int_fea}
\end{equation}
where 
$$ V \in \mathcal{M}_{n\times 2}, \,
V = \begin{pmatrix}
\bs{1}_{\mathcal{I}_1} & 0 \\ 
0 & \bs{1}_{\mathcal{I}_2} 
\end{pmatrix} \, .
$$
The spectral radius of $\frac{1}{\sqrt{n}}\begin{pmatrix}
A_{11} & A_{12} \\ 
A_{21} & A_{22}
\end{pmatrix}$ a.s. converges to $1$ due to the circular law \cite{tao_random_2010}. So as long as $\bs{s}_n$ is close to zero, the matrix $I - B_n$ is eventually (for large enough values of $n$) invertible.

\begin{theo}[Co-feasibility for the $2$-blocks model]
\label{th:fea_2b}
Assume that matrix $B_n$ is defined by the $2$-blocks model \eqref{eq:mat_int_fea}. Let $\bs{\beta} = (\beta_1,\beta_2), \, \beta_1+\beta_2 = 1$ represents the proportion of each community. Let $\bs{s}_n \xrightarrow[n\to\infty]{} 0$ and denote by $s_n^*=1/\sqrt{2\log n}$ the critical threshold. Let $x_n = (x_k)_{k\in[n]}$ be the
solution of \eqref{eq:fea_fixed}.
\begin{enumerate}
    \item If there exists $\varepsilon > 0$ such that eventually $\left \| (\bs{s_n}\circ \bs{s_n})\bs{\beta}^\top \right \|_{\infty} \geq (1+\varepsilon)(s_n^*)^2$ then
    $$
    \mathbb{P}\left \{ \underset{k\in [n]}{\min}\, x_k>0 \right \} \xrightarrow[n\rightarrow \infty]{} 0 \, .
    $$
    \item If there exists $\varepsilon > 0$ such that eventually $\left \| (\bs{s_n}\circ \bs{s_n})\bs{\beta}^\top \right \|_{\infty} \leq (1-\varepsilon)(s_n^*)^2$ then
    $$
    \mathbb{P}\left \{ \underset{k\in [n]}{\min}\, x_k>0 \right \} \xrightarrow[n\rightarrow \infty]{} 1 \, .
    $$
\end{enumerate}
\end{theo}

Following and adapting the ideas developed in \cite{bizeul_positive_2021}, this theorem could be proved mathematically in full detail. We rather focus on the main ideas and provide a sketch of proof in Appendix \ref{app:proof_feasibility_2b}. The extension to the $b$-blocks case can be found in Appendix \ref{app:extension_b_block}.

\begin{rem}
\label{rem:relaxed_assumptions_part3}
Proof of Theorem \ref{th:fea_2b} strongly depends on the assumption of Gaussianity and equal growth rates of each species. However, according to the approach of Bizeul \textit{et al.} \cite{bizeul_positive_2021}, these assumptions could be relaxed. In particular, the phenomenon seems to be universal, i.e. the feasibility threshold works for a wide range of distribution choices.
\end{rem}

In the critical regime $s \propto 1/\sqrt{\log(n)}$ or equivalently $s^{-1} \propto \sqrt{\log(n)}$. We thus introduce matrix $\bs{\kappa}$ defined by
$$
\bs{\kappa} = \frac{1}{\sqrt{\log(n)}} \begin{pmatrix}
s^{-1}_{11} & s^{-1}_{12}\\ 
s^{-1}_{21} & s^{-1}_{22}
\end{pmatrix} \,.
$$
Notice that at criticality $\bs{\kappa}$ will be of order $O(1)$. This will be convenient for ecological interpretations. Using the inequality of Theorem \ref{th:fea_2b}, the co-feasibility condition on $\bs{\kappa}$ writes
\begin{equation}
    \left \| (\bs{s_n}\circ \bs{s_n})\bs{\beta}^\top \right \|_{\infty} < (s_n^*)^2 \quad \Leftrightarrow \quad \max \left(\frac{2\beta_1}{\kappa_{11}^2}+\frac{2\beta_2}{\kappa_{12}^2},\frac{2\beta_1}{\kappa_{21}^2}+\frac{2\beta_2}{\kappa_{22}^2}\right) < 1 \,.
    \label{eq:cond_kappa}
\end{equation}
If for $i=1,2$, $\beta_i=\frac 12$ and the entry of the matrix $\bs{\kappa}$ are equal, then condition \eqref{eq:cond_kappa} gives the threshold $\kappa_{ij} > \sqrt{2}$, and we recover the same critical threshold $\sqrt{2\log(n)}$ as in \cite{bizeul_positive_2021}.
\newline

\begin{rem}
Assume $\kappa_{11} = \kappa_{22} = \nu_1$ and $\kappa_{12} = \kappa_{21} = \nu_2$, condition \eqref{eq:cond_kappa} is reformulated as:
$$
\max\left(\frac{2\beta_1}{\nu_1^2}+\frac{2\beta_2}{\nu_2^2},\frac{2\beta_1}{\nu_2^2}+\frac{2\beta_2}{\nu_1^2}\right)<1 \,.
$$
If $\beta_1,\beta_2$ and $\nu_2$ are fixed, then the phase transition on the intra-community interactions occurs at
$$
\nu_1 > \min\left(\sqrt{\frac{\beta_1}{\frac{1}{2}-\frac{\beta_2}{\nu_2^2}}},\sqrt{\frac{\beta_2}{\frac{1}{2}-\frac{\beta_1}{\nu_2^2}}}\right) \,.
$$
In Fig. \ref{fig:fea_transition}, the phase transition is represented for a selected set of parameters. Note that the transition is rather smooth. The threshold depends on $\nu_2$. Increasing $\nu_2$ (decreasing the inter-block interactions) lowers the co-feasibility threshold to at least 1 (for communities of the same size). 
\end{rem}
\begin{figure}[ht]
    \centering
    \includegraphics[width=0.8\textwidth]{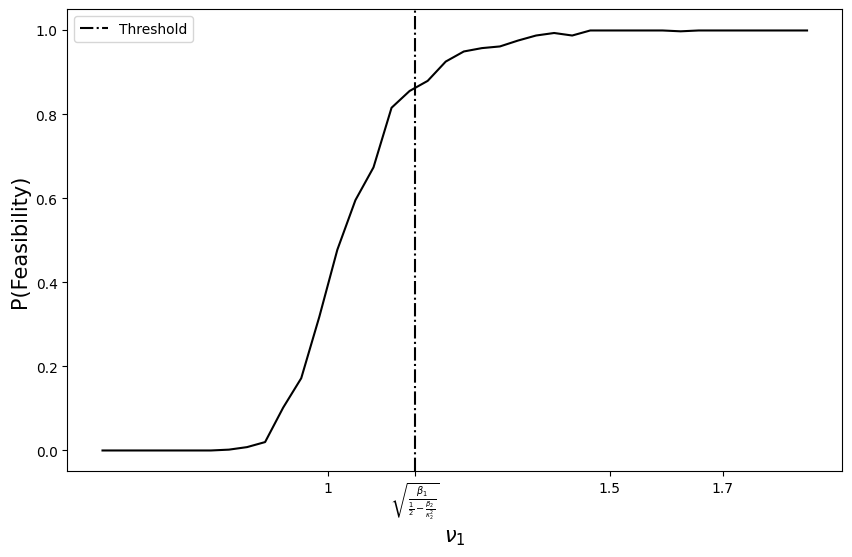}
    \caption{Transition towards co-feasibility for the $2$-blocks model \eqref{eq:mat_int}. For each value $\nu_1$ on the x-axis, we simulate $500$ matrices $B$ of size $n = 5000$ with two communities of the same size ($\beta_1=\beta_2=0.5$) with the inter-block interactions fixed at $s_{21}(\nu_2) = s_{21}(\nu_2) = 1/2\sqrt{\log(n)}$ and compute the solution $\bs{x}$ of Theorem \ref{th:fea_2b} at the scaling for the intra-block interactions $s_{11}(\nu_1) = s_{22}(\nu_1) = 1/\nu_1 \sqrt{\log(n)}$. The curve represents the proportion of feasible solutions $\bs{x}$ obtained for the 500 simulations. The dotdashed vertical line corresponds to $\nu_1 = \sqrt{\frac{\beta_1}{\frac{1}{2}-\frac{\beta_2}{\nu_2^2}}} = 2/\sqrt{3}$.
    } 
    \label{fig:fea_transition}
\end{figure}

\subsection{Preservation of co-feasibility}
\label{subsec:fea}
Equation \eqref{eq:cond_kappa} defines a ``co-feasibility domain" and gives a constraint in five dimensions. Recall that $\beta_1 = 1 - \beta_2$, the two communities of species can be studied independently i.e. the two components of equation \eqref{eq:cond_kappa} respectively give the feasibility condition for each community:
$$
\begin{cases}
 & \text{If } \frac{2\beta_1}{\kappa_{11}^2}+\frac{2\beta_2}{\kappa_{12}^2} < 1 \text{, then community 1 is feasible.} \\ 
 & \text{If } \frac{2\beta_1}{\kappa_{21}^2}+\frac{2\beta_2}{\kappa_{22}^2}<1  \text{, then community 2 is feasible.} 
\end{cases}
$$
The first community (resp. the second one) will be affected by changing $\kappa_{11}$, $\kappa_{12}$ (resp. $\kappa_{21}$, $\kappa_{22}$). In general, increasing the inter- or intra- interaction strength will decrease the probability of having a co-feasible equilibrium.

If $\kappa_{12}=\kappa_{21}=\infty$, then condition \eqref{eq:cond_kappa} gives the co-feasibility conditions for each community:
$$
s_{11} < \frac{1}{\sqrt{2\beta_1\log(n)}} \quad \textrm{and}\quad 
s_{22} < \frac{1}{\sqrt{2\beta_2\log(n)}} \, .
$$
For the same $s$, it means 
$$
s < \frac{1}{\sqrt{2\log(n)}\max(\beta_1,\beta_2)} \, .
$$

As an example of application, suppose we start with co-feasible communities of equal size ($\beta_1 = \beta_2 = 0.5$) and add interactions between these two groups, co-feasibility may be dropped (see Fig. \ref{fig:scenario2}). The co-feasibility domain is illustrated in Fig. \ref{fig:feasibility_domain}. It shows a threshold where the co-feasibility property is satisfied above the curve. This means that the lower the values of $\kappa_{11}$ and $\kappa_{22}$, i.e. the stronger the interactions within the groups, the more likely the co-feasibility property is lost. We can conclude that an independent group structure is more likely to be co-feasible and therefore stable, which supports previous work on compartmentalization models \cite{stouffer_compartmentalization_2011}.

\begin{figure}[ht]
    \centering
    \includegraphics[width = \textwidth]{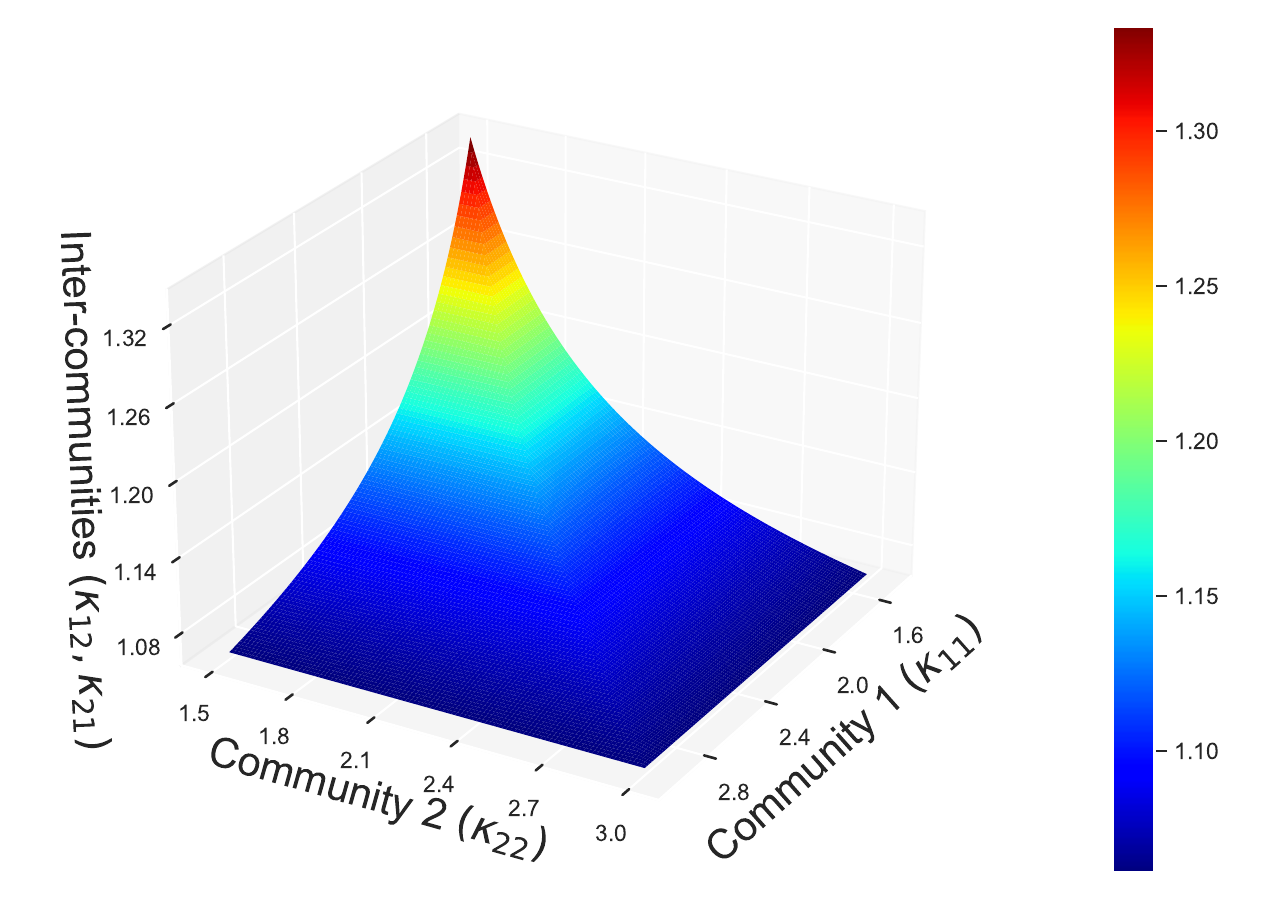}
    \caption{Representation of the co-feasibility phase diagram. The co-feasible domain is above the surface. The z-axis (resp x-axis) is the strength of interaction within community 1 - $\kappa_{11}$ (resp community 2 - $\kappa_{22}$). The y-axis is the inter-community interactions $\kappa_{12} = \kappa_{21}$. The colored area, where the gradient of color represents the strength of inter-community interactions (same values of the y-axis), illustrates the threshold between the co-feasible and non-co-feasible domains in the system \eqref{eq:LV2}.}
    \label{fig:feasibility_domain}
\end{figure}

\subsection{Impact of the community size}

For a fixed matrix $\bs{\kappa}$, the condition to have a co-feasible fixed point can be computed as a function of the size of each community i.e. the pair $\bs{\beta}=(\beta_1,\beta_2)$.
Starting from the co-feasibility inequality \eqref{eq:cond_kappa}:
$$
\max \left(\frac{2\beta_1}{\kappa_{11}^2}+\frac{2\beta_2}{\kappa_{12}^2},\frac{2\beta_1}{\kappa_{21}^2}+\frac{2\beta_2}{\kappa_{22}^2}\right) < 1\, ,
$$
the two components are studied independently,
\begin{equation*}
\frac{2\beta_1}{\kappa_{11}^2}+\frac{2(1-\beta_1)}{\kappa_{12}^2} < 1\quad \Leftrightarrow \quad \beta_1 \left(\frac{2}{\kappa_{11}^2}-\frac{2}{\kappa_{12}^2} \right) < 1-\frac{2}{\kappa_{12}^2}
\quad  \Rightarrow \quad \beta_1 < \frac{1-\frac{2}{\kappa_{12}^2}}{\left(\frac{2}{\kappa_{11}^2}-\frac{2}{\kappa_{12}^2} \right)}\quad \textrm{if}\ \kappa_{11} < \kappa_{12} \, .
\end{equation*}
Similarly, one has
\begin{equation*}
\sqrt{\frac{2\beta_1}{\kappa_{21}^2}+\frac{2(1-\beta_1)}{\kappa_{22}^2}} < 1 \quad \Rightarrow \quad \beta_1 >  \frac{1-\frac{2}{\kappa_{22}^2}}{\left(\frac{2}{\kappa_{21}^2}-\frac{2}{\kappa_{22}^2} \right)}\quad \textrm{if}\ \kappa_{22} < \kappa_{21} \, .
\end{equation*}
In the case where the intra-community interactions $(\kappa_{11}\,,\,\kappa_{22})$ are smaller than the inter-community interactions $(\kappa_{12}\,,\,\kappa_{21})$, we obtain an upper and a lower bound for the admissible size of each community $\beta_1,\beta_2$ to have a co-feasible equilibrium. In Fig. \ref{fig:impact_size}, different cases of the co-feasibility zone are represented according to the inter-community interactions $(\kappa_{12}\,,\,\kappa_{21})$. If the intra-community interactions are different, the community with the lowest interaction $\kappa_{ii}$ is advantaged i.e. the size of the community can be larger.

\begin{figure}[ht]
\centering
\begin{subfigure}[b]{0.32\textwidth}
  \centering
  \includegraphics[width=\textwidth]{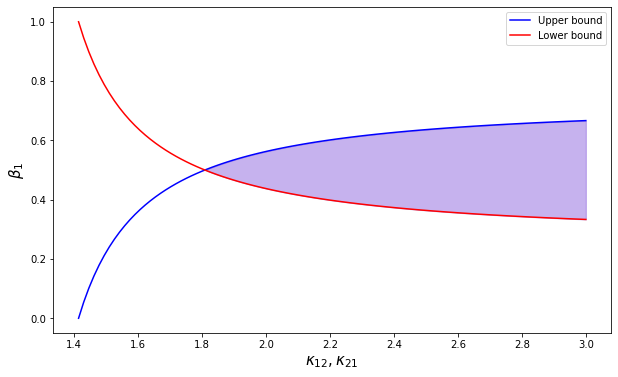}
  \caption{$\kappa_{11} = 1.2$, $\kappa_{22} = 1.2$}

\end{subfigure}
\hfill
\begin{subfigure}[b]{0.32\textwidth}
  \centering
  \includegraphics[width=\textwidth]{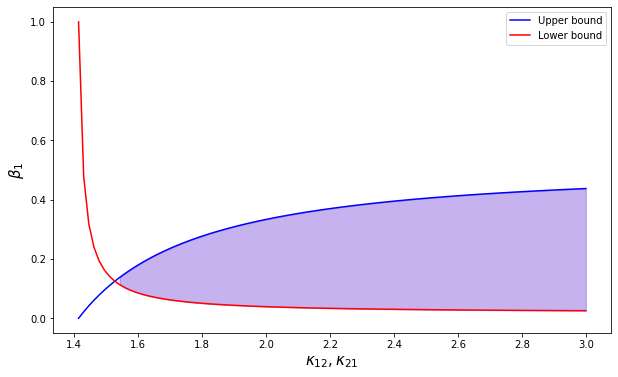}
  \caption{$\kappa_{11} = 1.2$, $\kappa_{22} = 1.4$}

\end{subfigure}
\hfill
\begin{subfigure}[b]{0.32\textwidth}
  \centering
  \includegraphics[width=\textwidth]{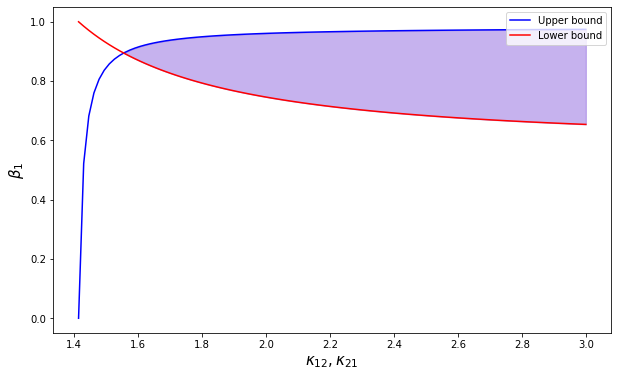}
  \caption{$\kappa_{11} = 1.4$, $\kappa_{22} = 0.9$}
\end{subfigure}
\caption{Representation of the co-feasibility domain depending on the fixed intra-community interaction. In (a), (b), (c), a different scenario of intra-community interaction is presented. Each panel represents the upper-bound (blue curve) and the lower-bound (red curve) of the size of community 1 as a function of the interaction between two communities $(\kappa_{12},\kappa_{21})$. The blue area is the admissible zone to have a co-feasible fixed point in \eqref{eq:LV2}. The size of community $2$ is equal to $\beta_2 = 1-\beta_1$.
}
\label{fig:impact_size}
\end{figure}

\begin{rem}
If $(\kappa_{11}, \kappa_{22})$ are greater than $(\kappa_{12}, \kappa_{21})$, there may be different situations depending on the value of $(\kappa_{11}, \kappa_{22})$, i.e. if $(\kappa_{11}, \kappa_{22})> \sqrt{2}$ are large, we can have co-feasibility (see Fig.\ref{fig:kappa_big}), whereas if $(\kappa_{11}, \kappa_{22})$ are small, we may not have co-feasibility (see Fig. \ref{fig:kappa_small}). From a biological perspective, we believe that the case $(\kappa_{11}, \kappa_{22})$ smaller than $(\kappa_{12}, \kappa_{21})$ is more significant because the interactions within each community are stronger than those between communities.  \end{rem}

\begin{figure}[ht]
\centering
\begin{subfigure}[b]{0.49\textwidth}
  \centering
    \includegraphics[width=\textwidth]{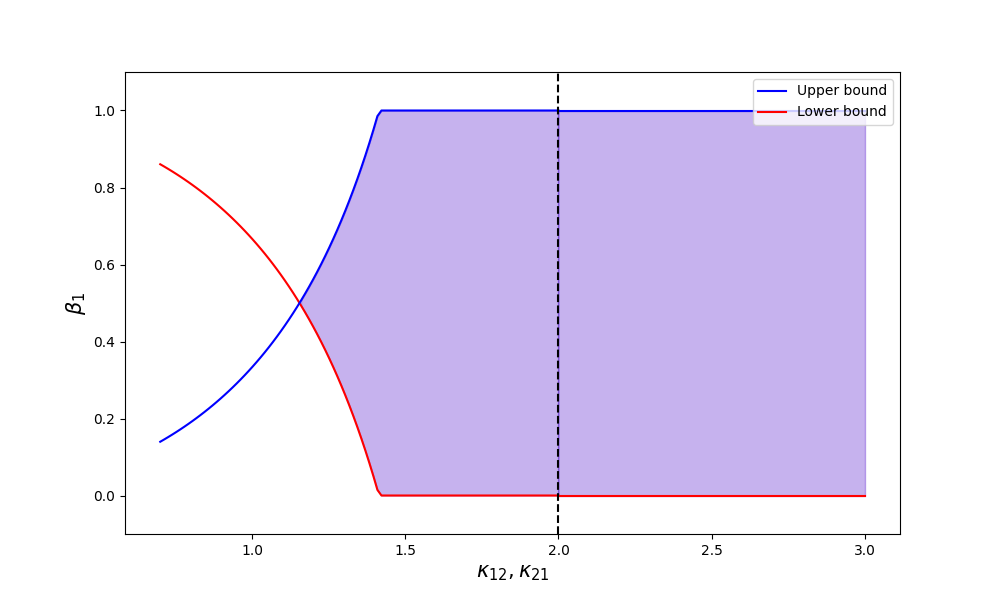}
  \caption{$\kappa_{11} = 2$, $\kappa_{22} = 2$}
  \label{fig:kappa_big}
\end{subfigure}
\hfill
\begin{subfigure}[b]{0.49\textwidth}
  \centering
  \includegraphics[width=\textwidth]{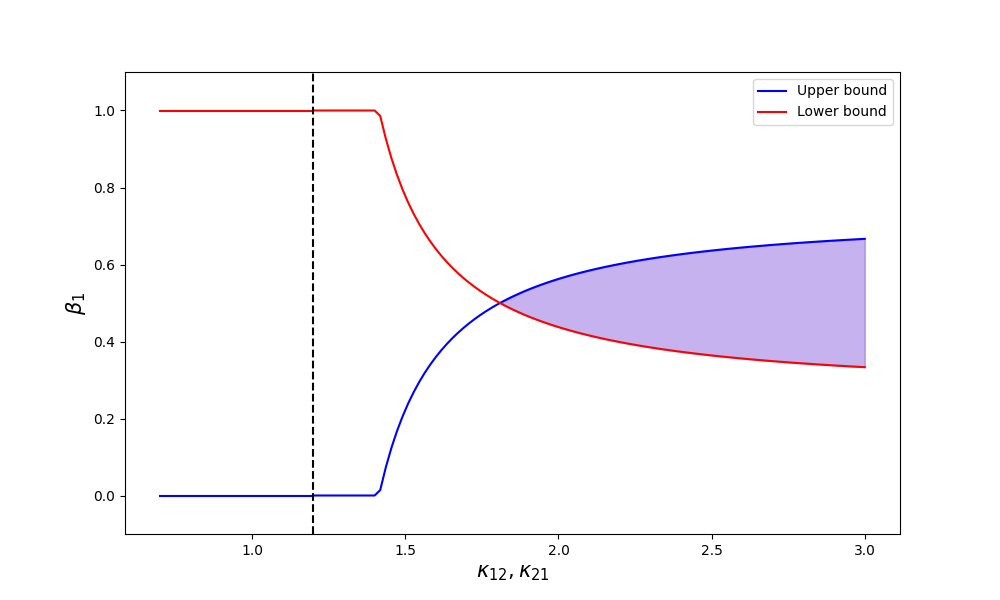}
  \caption{$\kappa_{11} = 1.2$, $\kappa_{22} = 1.2$}
  \label{fig:kappa_small}
\end{subfigure}
\caption{Representation of the co-feasibility domain depending on the fixed intra-community interaction in the case where $(\kappa_{11},\kappa_{22})$ can be smaller than $(\kappa_{12},\kappa_{21})$. Each panel represents the upper-bound (blue curve) and the lower-bound (red curve) of the size of community 1 as a function of the interaction between two communities $(\kappa_{12},\kappa_{21})$. The blue area is the admissible zone to have a co-feasible fixed point in \eqref{eq:LV2}. The dashed line represents the threshold value of $(\kappa_{11},\kappa_{22})$. The size of community $2$ is equal to $\beta_2 = 1-\beta_1$. In (a), we present a scenario of intra-community interaction where the values of $(\kappa_{11},\kappa_{22}) = 2$ are large. In (b), we present a scenario of intra-community interaction where the values of $(\kappa_{11},\kappa_{22}) = 1.2$ are small.
}
\end{figure}

\subsection{Connection increases co-feasibility}

In Section \ref{subsec:fea}, we analyzed the co-feasibility condition for a scenario involving two communities. We presented a co-feasibility domain defined by
\begin{align}
&\quad \left \| (\bs{s_n}\circ\bs{s_n}) \bs{\beta}^\top \right \|_{\infty} < \frac{1}{2\log(n)} := (s_n^*)^2\, , \nonumber \\
&\Leftrightarrow \quad \max \left(\beta_1s_{11}^2+\beta_2s_{12}^2,\beta_1s_{21}^2+\beta_2s_{22}^2\right) < \frac{1}{2\log(n)}\, , \nonumber \\
&\Leftrightarrow \quad \beta_1s_{11}^2+\beta_2s_{12}^2 < \frac{1}{2\log(n)} \quad \text{and} \quad \beta_1s_{21}^2+\beta_2s_{22}^2 < \frac{1}{2\log(n)} \label{eq:constraint_community}\, .
\end{align}
These two distinct conditions within the communities have led us to the conclusion that community isolation is beneficial for coexistence. However, general constraints that affect all interactions could be considered at the ecosystem scale. To this end, we introduce a complementary condition: the global variance of the interaction coefficients in the ecosystem remains invariant, i.e.
$$
\beta_1s_{11}^2+\beta_2s_{12}^2+\beta_1s_{21}^2+\beta_2s_{22}^2 = \Gamma \, ,
$$
with $\Gamma > 0$.

In order to remove the $n$ dependency, we define $\gamma_{ij}^2 = s_{ij}^2\log(n)$ and $\Gamma_n = \Gamma\log(n)$ (note that the $\gamma$ coefficients are defined differently from the $\kappa$ coefficients). Combined with condition \eqref{eq:constraint_community}, we get the following system of equations:
\begin{equation}
\label{eq:cond_fea}
\begin{cases}
&\beta_1\gamma_{11}^2+\beta_2\gamma_{12}^2+\beta_1\gamma_{21}^2+\beta_2\gamma_{22}^2 = \Gamma_n \, ,\\
& 2\beta_1\gamma_{11}^2+2\beta_2\gamma_{12}^2 < 1 , \quad \text{(feasibility condition for community 1)}\, , \\
& 2\beta_1\gamma_{21}^2+2\beta_2\gamma_{22}^2 < 1 , \quad \text{(feasibility condition for community 2)}\, .
\end{cases}
\end{equation}
Assuming the fixed intra-community variances $\gamma_{11}$ and $\gamma_{22}$, we seek to determine the co-feasibility conditions for the inter-community variances $\gamma_{12}$ and $\gamma_{21}$. In this case, the constraint on the total variance corresponds to the equation of an ellipse in the $(\gamma_{12},\gamma_{21})$ plane:
\begin{align*}
& \quad \quad \beta_1\gamma_{11}^2+\beta_2\gamma_{12}^2+\beta_1\gamma_{21}^2+\beta_2\gamma_{22}^2 = \Gamma_n \, , \\
&\Leftrightarrow \left(\frac{\beta_2}{\Gamma_n-\beta_1\gamma_{11}^2-\beta_2\gamma_{22}^2} \right)\gamma_{12}^2+\left(\frac{\beta_1}{\Gamma_n-\beta_1\gamma_{11}^2-\beta_2\gamma_{22}^2} \right)\gamma_{21}^2 = 1 \, ,
\end{align*}
The equations below provide the values for the semi-major axis $a$ and the semi-minor axis $b$ of the ellipse:
$$
a = \sqrt{\frac{\Gamma_n-\beta_1\gamma_{11}^2-\beta_2\gamma_{22}^2}{\beta_2}} \quad , \quad b = \sqrt{\frac{\Gamma_n-\beta_1\gamma_{11}^2-\beta_2\gamma_{22}^2}{\beta_1}} \,.
$$
Note that if both communities are of equal size ($\beta_1 = \beta_2$), a circle with radius $a$ is obtained.
\newline

From a visual standpoint, the conditions \eqref{eq:cond_fea} are depicted in Figure \ref{fig:fea_constraint}. Since the coefficients $(\gamma_{12},\gamma_{21})$ are non-negative, we are only interested in the positive orthant. The feasibility condition for community 1 is given by the horizontal axis defined by
$$
2\beta_1\gamma_{11}^2+2\beta_2\gamma_{12}^2 < 1 \quad \Leftrightarrow \quad \gamma_{12} < \sqrt{\frac{1-2\beta_1\gamma_{11}^2}{2\beta_2}} \, ,
$$
and the one of community 2 is given by the vertical axis defined by
$$
2\beta_1\gamma_{21}^2+2\beta_2\gamma_{22}^2 < 1 \quad \Leftrightarrow \quad \gamma_{21} < \sqrt{\frac{1-2\beta_2\gamma_{22}^2}{2\beta_1}} \, .
$$

The intersection between the vertical (resp. horizontal) line and the ellipse occurs when the semi-major axis (resp. semi-minor axis) exceeds the vertical condition $a > \sqrt{(1-2\beta_1\gamma_{11}^2)/(2\beta_2)}$ (resp. horizontal condition $b > \sqrt{(1-2\beta_2\gamma_{22}^2)/(2\beta_1)}$).

\begin{rem}
By replacing the feasibility condition of $\gamma_{21}$ in the equation of the ellipse, we can derive the intersection between the vertical axis and the ellipse as follows: 
$$
\gamma_{21}^2 = \frac{\Gamma_n-\beta_2\gamma_{22}^2-1/2}{\beta_1} \, ,
$$
equivalent to the feasibility condition of $\gamma_{12}$ (by replacing $\Gamma_n$):
$$
\gamma_{12} < \sqrt{\frac{1-2\beta_1\gamma_{11}^2}{2\beta_2}} \, .
$$
We identify the range of co-feasibility between the two groups for $\gamma_{21}$
$$
\frac{\Gamma_n-\beta_2\gamma_{22}^2-1/2}{\beta_1} < \gamma_{21}^2 < \frac{1-2\beta_2\gamma_{22}^2}{2\beta_1} \, .
$$
\end{rem}

This simple framework allows for testing different scenarios. Figure \ref{subfig:ellipse1} is the reference figure. It represents a situation where all the potential interactions between communities lead to co-feasibility. In Figure \ref{subfig:ellipse2}, the total system variance is increased, which results in a reduction of co-feasibility options: the interactions between the two communities must be high $\gamma_{12}, \gamma_{21}\gg 0$. When the intra-community variances are increased as shown in Figure \ref{subfig:ellipse3}, the ellipse shrinks, and the set of inter-community interaction variances is reduced. This observation reinforces the findings of Section \ref{subsec:fea} where large interaction between communities enables co-feasibility through community isolation. In the concluding example in Figure \ref{subfig:ellipse4}, we reduce only the $\gamma_{11}$ interaction in community 1. We observe that the impact of community 1 on 2 $\gamma_{21}$ must be weaker, but the impact of community 2 on 1 $\gamma_{12}$ can no longer be weak. Weaker interactions within community 1 imply a stronger connection between the communities for co-feasibility. 
\begin{figure}[ht]
\centering
\begin{subfigure}{0.45\textwidth}
  \includegraphics[width=\textwidth]{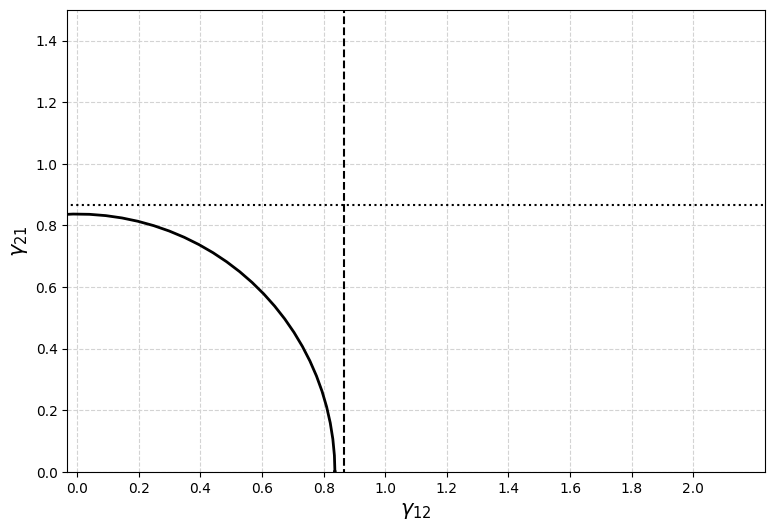}
  \caption{$\gamma_{11} = \gamma_{22} = 0.5,\, \Gamma_n = 0.6$.}
  \label{subfig:ellipse1}
\end{subfigure}
\hfill
\begin{subfigure}{0.45\textwidth}
\includegraphics[width=\textwidth]{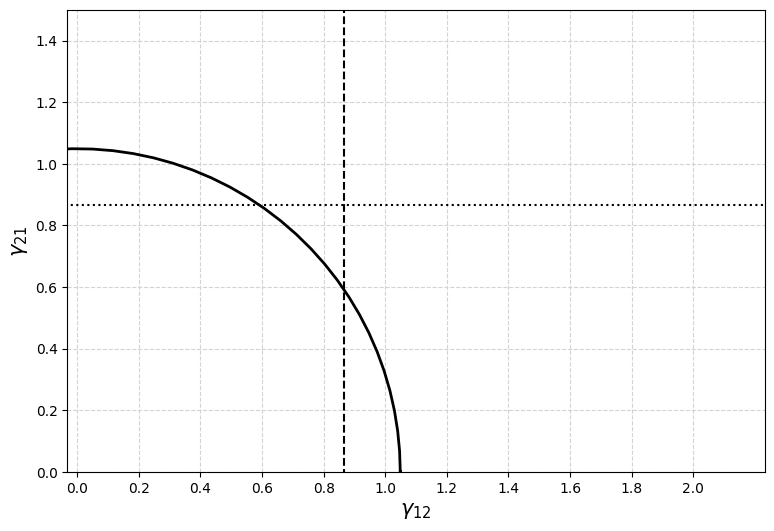}
\caption{$\gamma_{11} = \gamma_{22} = 0.5,\, \Gamma_n = 0.8$.}
\label{subfig:ellipse2}
\end{subfigure}
\hfill
\begin{subfigure}{0.45\textwidth}
  \includegraphics[width=\textwidth]{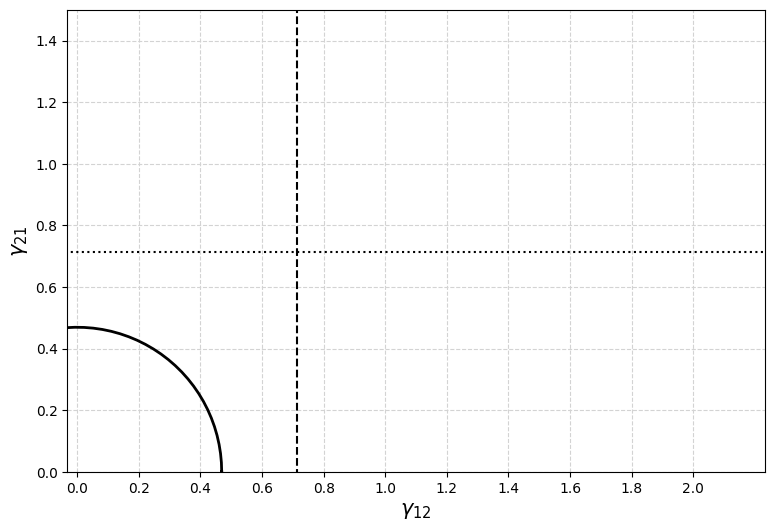}
  \caption{$\gamma_{11} = \gamma_{22} = 0.7,\, \Gamma_n = 0.6$.}
  \label{subfig:ellipse3}
\end{subfigure}
\hfill
\begin{subfigure}{0.45\textwidth}
  \includegraphics[width=\textwidth]{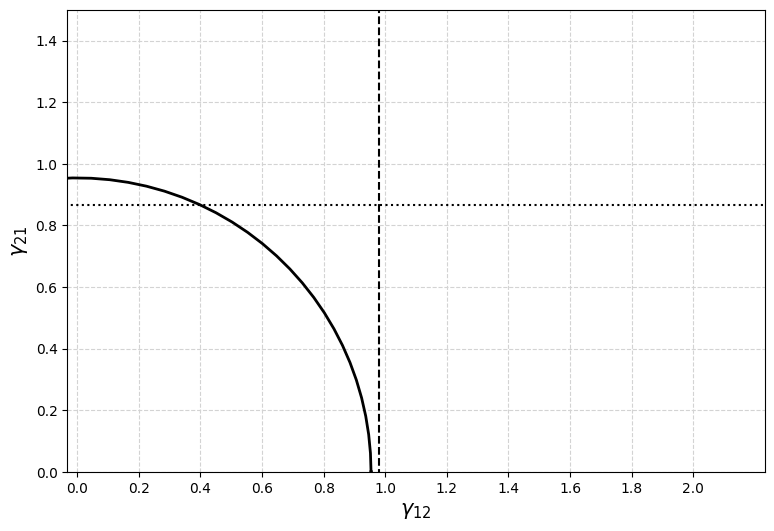}
  \caption{$\gamma_{11} = 0.2,\, \gamma_{22} = 0.5,\, \Gamma_n = 0.6$.}
  \label{subfig:ellipse4}
\end{subfigure}
\caption{Graphical representation of equation system \eqref{eq:cond_fea}. The solid line circle constrains the total variance of the system (equation 1 in \eqref{eq:cond_fea}), while the dashed and dotted lines correspond to feasibility conditions for community 1 (equation 2 in \eqref{eq:cond_fea}) and 2 (equation 3 in \eqref{eq:cond_fea}), respectively. In each figure, two communities of equal size are considered with $\beta_1 = \beta_2 = 1/2$, creating the solid line circle (rather than an ellipse) in this particular case. Each figure illustrates a distinct interaction situation between two communities outlined in the caption. Figure (a) serves as the reference against which the other figures (b)-(c)-(d) are compared. The co-feasibility arises only for the values of ($\gamma_{12},\gamma_{21}$) on the line found within the inner square at the bottom left-hand corner of the figure. When values of $(\gamma_{12},\gamma_{21})$ on the line are outside this inner square, then one or the other community is not feasible anymore.}
\label{fig:fea_constraint}
\end{figure}

\section{Discussion and perspectives}
In this paper, we described a model of the dynamics of species abundances when the interaction among species is structured in multiple communities. The main interest is to outline the effect of a block structure on the stability and persistence of species. We defined an interaction matrix per block which has several characteristics such as the strength of the interactions $\bs{s}$ and the size of the community $\bs{\beta}$. Specifically, we described the dynamics and properties of each community in the system (feasibility, proportion of surviving species, mean and root mean square of the abundances of surviving species) and their effect on each other. In this context, we focused most of our analysis on the case of two interacting communities. However, our results can be extended to more than $2$ communities (see Appendix \ref{app:extension_b_block}).

First, theoretical conditions were given for a unique globally stable equilibrium in the model \eqref{eq:LV2} with surviving and vanishing species. This follows from Lyapunov conditions related to a result of Takeuchi and Adachi \cite{takeuchi_existence_1980} and random matrix theory. These stability results had been found in the case of a single community by Clenet \textit{et al.} \cite{clenet_equilibrium_2023}. This complements the stability properties  in the Lotka-Volterra system studied by Stone \cite{stone_feasibility_2018} and Gibbs \textit{et al.} \cite{gibbs_effect_2018}. Recent random matrix methods allow us to describe the spectrum of a block matrix and plot it numerically. For a detailed discussion of random matrices in the Lotka-Volterra model, see Akjouj \textit{et al.} \cite{akjouj_complex_2022}.

Subsequently, we gave heuristics on the surviving species (proportion, mean and root mean square of their abundances). These heuristics have also been found in the case of a single community by Clenet \textit{et al.} \cite{clenet_equilibrium_2023}. From a physicist's point of view and using the methods of Bunin \cite{bunin_ecological_2017} and Galla \cite{galla_dynamically_2018}, Barbier \textit{et al.} \cite{barbier_generic_2018} and Poley \textit{et al.} \cite{poley_generalized_2023} have extended the heuristics in the block and cascade model. Previously, obtaining properties on surviving species in the LV model (not normalized by $\sqrt{n}$) was already done by Servan \textit{et al.} \cite{servan_coexistence_2018} where they consider a different growth rate for each species. The study of the stability and properties of surviving species in the LV system has also been carried out by Pettersson \textit{et al.} \cite{pettersson_predicting_2020, pettersson_stability_2020}. From an ecological point of view, heuristics are derived from the properties of interactions between multiple communities. 

In a third part, we studied the condition under which the feasibility threshold exists where all species coexist. We extend the feasibility results found by Bizeul and Najim \cite{bizeul_positive_2021} in the case of a block structure. A co-feasibility threshold was found in the form of an inequality that must be verified to have a feasible community set. This complements the recent results on interactions with a sparse structure \cite{akjouj_feasibility_2022} and interactions with a correlation profile \cite{clenet_equilibrium_2022}. We notice that to maximize the probability of co-feasibility, we need to minimize the interactions between the communities. Additionally, a community with weaker interactions can exhibit a larger total abundance in the ecosystem while maintaining the co-feasibility threshold. At the ecosystem level, when a generic constraint that affects all interactions is added, weaker interactions within one of the communities suggest a stronger connection between the communities for co-feasibility.

\vspace{0.2cm}

There are still many mathematical and ecological questions that remain unanswered in this type of model.

First, a rigorous mathematical proof of the heuristics presented here would be of interest, although the LCP procedure induces an a priori statistical bias that is difficult to handle. This issue is still pending in the single community case \cite{clenet_equilibrium_2023} and appears to be challenging to address. Recently, Akjouj \textit{et al.} \cite{akjouj_equilibria_2023} provided a rigorous proof using an approximate message passing (AMP) approach in the single community model with an interaction matrix taken from the Gaussian Orthogonal Ensemble (GOE). Their approach was based on work by Hachem \cite{hachem_approximate_2023}.

Second, we could extend the heuristics for two different scenarios. On the one hand, it would be interesting to add pairwise correlations between species  coefficients $A_{ij}$. This has already been done by physicists, see \cite{barbier_generic_2018,poley_generalized_2023}. In the study of feasibility, it was shown that a correlation profile does not change the feasibility threshold \cite{clenet_equilibrium_2022}. On the other hand, for the sake of simplicity, we have chosen to set the growth rates equal to the same value $r_k = 1$ for $k \in [n]$. It would be relevant to control the distribution of the growth rate as in \cite{servan_coexistence_2018} or to consider structural stability as in Saavedra \textit{et al.} \cite{saavedra_structural_2017}, i.e. how much can the growth rates be perturbed (initially all equal to $1$) without changing the type of equilibrium $\bs{x}^*$ obtained.
 
There are many applications of this kind of models in ecology. We could consider a spatial structure that accounts for spatial proximity in the sense that two nearby communities tend to be more strongly connected. For example, in an aquatic environment, we could imagine the existence of an up/down gradient in a water column. Fig. \ref{fig:aquatic_scheme} illustrates a situation where three communities are involved.

\begin{figure}[ht]
\centering
\begin{subfigure}[b]{0.48\textwidth}
    \includegraphics[width=\textwidth]{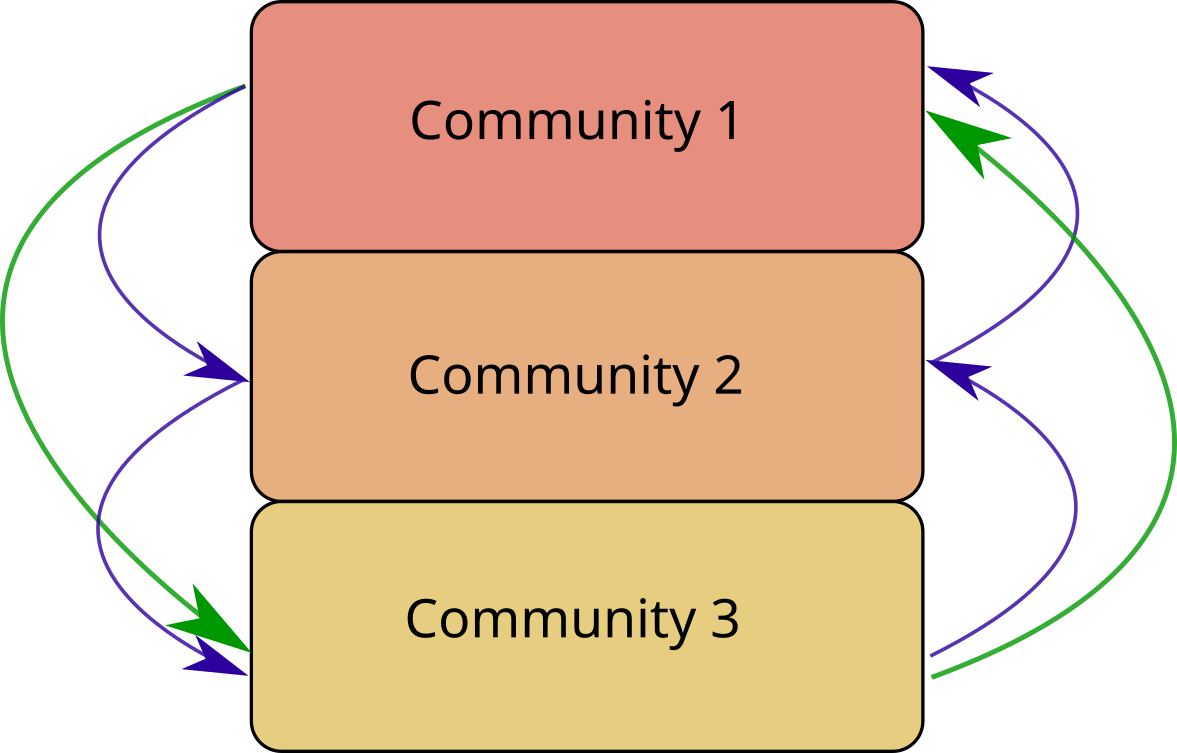}
    \caption{}
    
  \end{subfigure}%
  \hfill   
  \begin{subfigure}[b]{0.48\textwidth}
    \includegraphics[width=\textwidth]{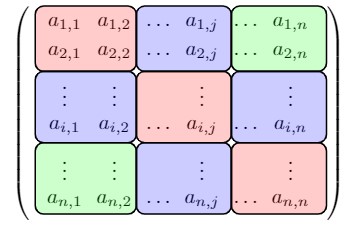}
    \caption{}
    
  \end{subfigure}%
\caption[Illustration of the perspective in the block model framework.]{In (a), a representation of the gradient of interaction between three communities in a water column is represented. The blue arrows correspond to strong interaction strength due to their spatial proximity. On the opposite, the communities 1 and 3 are separated, the green arrow represents a weaker interaction. In (b), the block matrix associated with this type of model is displayed. The colors of the blocks corresponds to the colors of the arrows. The red colored block corresponds to intra-community interactions.}
\label{fig:aquatic_scheme}
\end{figure}

Originally introduced by R.T. Paine \cite{paine_food_1966,paine_pisaster-tegula_1969}, the concept of keystone species is widely used in ecology i.e. one species controls the coexistence of the others and species are lost after the removal of this keystone species. Mouquet \textit{et al.} \cite{mouquet_extending_2013} suggested extending the concept of keystone species to communities. In the block system, one could analyze the existence of a keystone community that would have disproportionately large effect on other communities. In a metacommunity dynamic, Resetarits \textit{et al.} \cite{resetarits_testing_2018} have explored the concept of keystone communities, where some patches have stronger effect on others. 

One could imagine that the same species is present several times in the system, but in different blocks, see Gravel \textit{et al.} \cite{gravel_stability_2016}. In this case, the inter-blocks represent interactions between spatially isolated communities (so should be less strong). If each diagonal or non-diagonal block is a copy of the same interaction pattern (possibly slightly perturbed) and we can add linear effects to the system to represent emigration and immigration, then we could study the feasibility properties of this system. In \cite{gravel_stability_2016}, they found that stability is most likely when dispersal (which controls off-diagonal blocks) is intermediate.

Last but not least, it would be relevant to compare the patterns obtained with data in ecology, as in the recent article by Hu \textit{et al.} \cite{hu_emergent_2022} in the case of a single community.

\section*{Script and code availability}
Script, codes and figures are available online on Github \cite{code}. The code is written in Python.
\newline
\url{https://github.com/maxime-clenet/Impact-of-a-block-structure-on-the-Lotka-Volterra-model}

\section*{Fundings}
M.C., F.M. and J.N. are supported by the CNRS 80 prime project LotKA-VolterRA models: when random maTrix theory meets theoretical Ecology (KARATE).

\section*{Conflict of interest disclosure}
The authors of this preprint declare that they have no financial conflict of interest with the content of this article.
\vfill\pagebreak

\bibliographystyle{alpha}
\bibliography{references}

\newpage

\appendix

\section{Stieljes transform}
\label{app:stieljes}
We provide some reminders regarding Stieltjes transforms, a central element of proofs in random matrix theory. We denote by
$$
\mathbb{C}^{+} := \{ z\in \mathbb{C} \, : \, \mathrm{Im}(z) > 0 \} 
$$
the upper half of the complex plane.
\begin{defi}[Stieltjes transform]
Let $\nu \in \mathcal{P}(\mathbb{R})$ be a probability measure. The Stieltjes transform of $\nu$, denoted by $g_\nu:\mathbb{C}^{+}\rightarrow \mathbb{C}$, is defined by
$$
g_{\nu}(z) = \int \frac{1}{\lambda-z}\nu(d\lambda) \, , z \in \mathbb{C^+} \, .
$$
\end{defi}

\begin{rem}
\label{rem:trace_resolvent}
Let $\nu_X$ be the empirical measure of the eigenvalues $\lambda_1(X),\cdots,\lambda_n(X)$ of the symmetric matrix $X$ $\in \mathcal{M}_n(\mathbb{C})$ define by
\begin{equation*}
    \nu_X := \frac{1}{n}\sum_{k=1}^{n}\delta_{\lambda_k(X)}.
\end{equation*}
then the associated Stieltjes transform is given by
$$
g_{\nu_X}(z) = \int \frac{1}{\lambda-z}\nu_X(d\lambda) = \frac{1}{n} \sum_{i=1}^n \frac{1}{\lambda_i-z} = \frac{1}{n}\mathrm{Tr}\left((X-zI)^{-1}\right) \, ,
$$
where $Q = (X-zI)^{-1}$ is the resolvent of the matrix $X$ and $\mathrm{Tr}(Q)$ is the trace of matrix $Q$.
\end{rem}
\begin{prop}[Stieltjes inversion]
\label{prop:block_stielt_inv}
Let $g_{\nu}$ the Stieltjes transform of the measure $\nu$ of finite mass $\nu(\mathbb{R})$. If $a,b \in \mathbb{R}$ and $\nu(\{a\}) = \nu(\{b\}) = 0$, then
\begin{eqnarray*}
\nu(a,b) &=&  \frac{1}{\pi} \underset{y \rightarrow 0^+}{\lim} \mathrm{Im} \int_a^b g_{\nu}(x+iy)dx \, ,\\
\forall x\in \mathbb{R},\quad \nu(\{x \}) &=& \frac{1}{\pi} \underset{y \rightarrow 0^+}{\lim} \mathrm{Im} (g_{\nu}(x+iy)) \, .
\end{eqnarray*}
\end{prop}

\section{Numerical methods}

\subsection{Methods for validating heuristics \ref{heur:properties_surviving_species}}

\label{app:num}
To verify the system of equations of heuristics \ref{heur:properties_surviving_species}, the simulations on the properties of surviving species are performed in two distinct methods (see Fig. \ref{fig:block_test}). On the one hand, we use a standard solver (cf. scipy.optimize) to find the theoretical solutions by finding a local minimum of the system of equations (a modification of the Powell hybrid method). On the other hand, we simulate a large number of matrix $B$, each corresponding to an experiment, and we resolve the associated LCP problem using the Lemke's algorithm (see the lemkelcp package \cite{packagelemke}). The empirical solutions are computed using a Monte Carlo experiment, i.e. we use the LCP solution to compute the properties of the surviving species and we make an average over the ensemble of experiments. As a baseline, the dynamics of Lotka-Volterra are achieved by a Runge-Kutta method of order 4 (RK4) implemented in the code.

\subsection{Spectrum: a computer based approach.}

Theorem \ref{th:block_unicite_centered} only provides sufficient conditions for the existence of a unique stable equilibrium and is based on the rough asymptotic upper bound estimation $\Sigma = 2\left \| S \right \|_\infty^{1/2}$. We can assess the sharpness of this bound by comparing it to the limiting spectrum of matrix $H$, which can be plotted via numerical simulations. An efficient way to compute numerically the spectrum of the matrix $H$ comes from the system of non linear equations \eqref{eq:QVE}.

Starting from the QVE equations \eqref{eq:QVE} associated to the matrix $H$, the system takes the simpler form
\begin{equation*}
\begin{cases}
 -\frac{1}{m(z)} &=z + 2\beta_1s_{11}^2 m(z) + \beta_2(s_{12}^2+s_{21}^2)\check m(z)   \\ 
 -\frac{1}{\check m(z)} &=z + \beta_1(s_{12}^2+s_{21}^2)m(z) + 2\beta_2s_{22}^2\check m(z) 
\end{cases}\, ,
\end{equation*}
where for $k \in \mathcal{I}_1$, $m_k(z) = m(z)$ and for $k \in \mathcal{I}_2$, $m_k(z) = \check m(z)$. 
One can indeed check that with $m$ and $\check m$ defined as above, then 
\begin{equation}\label{eq:m-special-form}
\bs{m}^\top=[\underbrace{m,\cdots,m}_{n_1}\,,\, \underbrace{\check m,\cdots,\check m}_{n_2}] 
\end{equation}
satisfies the QVE equations \eqref{eq:QVE}; here $n_1=|{\mathcal I}_1|$ and $n_2=|{\mathcal I}_2|$. Since this solution $\bs{m}$ is unique, $\bs{m}$ is necessarily given by the simplified form \eqref{eq:m-special-form}.

All the knowledge of equation \eqref{eq:QVE} hence relies on functions $m(z)$ and $\check m(z)$.
Then, using the RMT results developed in \cite{ajanki_universality_2017}, the resolvent $G$ of the symmetric matrix $H$ can be approximated by
\begin{equation*}
    G(z) = (H-zI)^{-1} \simeq \mathrm{diag}(m(z)\bs{1}_{\mathcal{I}_1}^{\top},\check m(z)\bs{1}_{\mathcal{I}_2}^{\top}) \,.
\end{equation*}
From Remark \ref{rem:trace_resolvent}, the trace of the resolvent is equal to the Stieltjes transform 
$$
g(z) = \frac{1}{n}\mathrm{Tr}(G) \simeq \beta_1 m(z) + \beta_2 \check m(z)
$$ 
of the spectral measure. Finally, the spectral density can be obtained using the Stieltjes inversion formula, see Prop. \ref{prop:block_stielt_inv}. The spectral density of the matrix $H$ can be computed numerically by an iterative scheme. The initial condition of the two measurements $(m,\check m)$ is $m_0 = \check m_0 =  -\frac{1}{z}$. Then, the iterative scheme 
\begin{equation*}
\begin{cases}
 -\frac{1}{m_p} &=z + 2\beta_1s_{11}^2 m_{p-1} + \beta_2(s_{12}^2+s_{21}^2)\check m_{p-1} \\ 
 -\frac{1}{\check m_p} &=z + \beta_1(s_{12}^2+s_{21}^2)m_{p-1} + 2\beta_2s_{22}^2\check m_{p-1}
\end{cases}\, ,
\end{equation*}
converge to 
$m_{\infty} = \underset{p \rightarrow +\infty}{\lim} m_p$ and $\check m_{\infty} = \underset{p \rightarrow +\infty}{\lim} \check m_p$. The last step consists in using the Stieltjes inversion formule, see Prop.\ref{prop:block_stielt_inv}. 
\begin{rem}
    To handle the Stieltjes inversion (Prop.\ref{prop:block_stielt_inv}) numerically,  it is similar as starting with $z = x+\epsilon i,  \epsilon \approx 10^{-3}$.
\end{rem}
In Fig. \ref{fig:semi-circle}, we present the numerical estimation of the spectral density for different types of interactions of the matrix $H$. 
\begin{figure}[ht]
\centering
\begin{subfigure}{0.48\textwidth}
  \includegraphics[width=\textwidth]{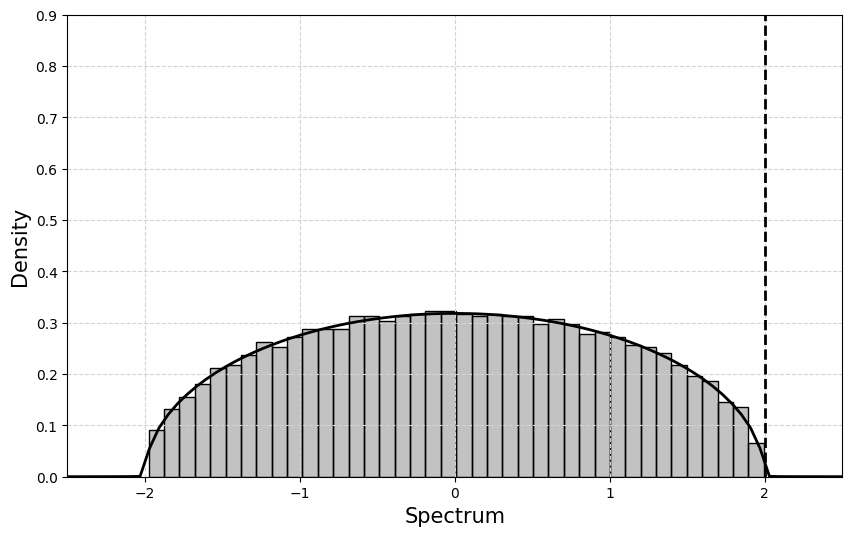}
  \caption{$\bs{\beta} = [1/2, 1/2]$, $\bs{s} = \begin{pmatrix}
1/\sqrt{2} & 1/\sqrt{2} \\ 
1/\sqrt{2} & 1/\sqrt{2}
\end{pmatrix}$
}
\label{subfig:pattern_1}
\end{subfigure}
\hfill
\begin{subfigure}{0.48\textwidth}
  \includegraphics[width=\textwidth]{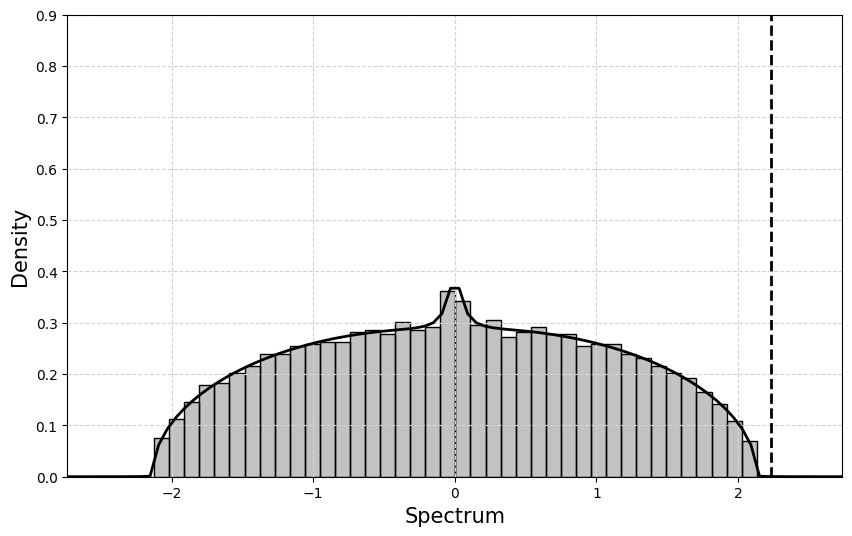}
    \caption{$\bs{\beta} = [1/2, 1/2]$, $\bs{s} = \begin{pmatrix}
1/2 & 1 \\ 
1 & 1/8
\end{pmatrix}$
}
\label{subfig:pattern_2}
\end{subfigure}
\hfill
\begin{subfigure}{0.48\textwidth}
  \includegraphics[width=\textwidth]{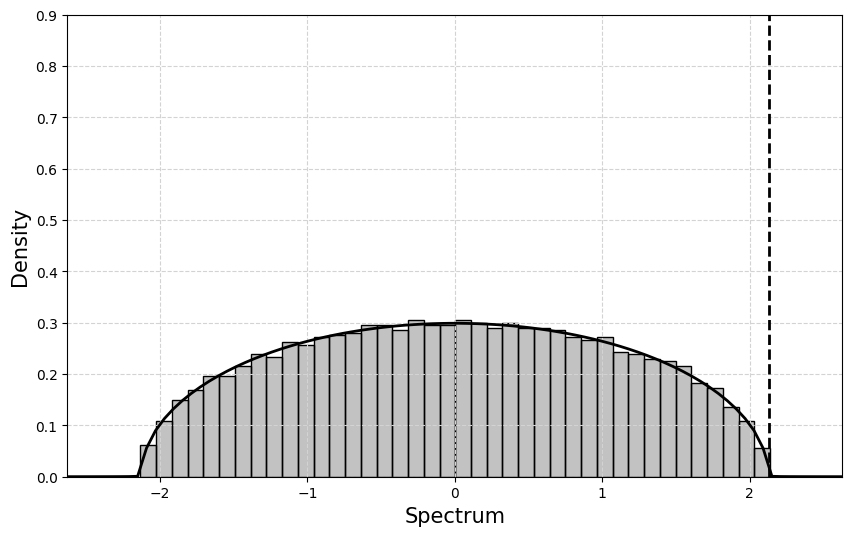}
  \caption{$\bs{\beta} = [1/2, 1/2]$, $\bs{s} = \begin{pmatrix}
1 & 1/2 \\ 
1/8 & 1
\end{pmatrix}$
}
\label{subfig:pattern_3}
\end{subfigure}
\hfill
\begin{subfigure}{0.48\textwidth}
  \includegraphics[width=\textwidth]{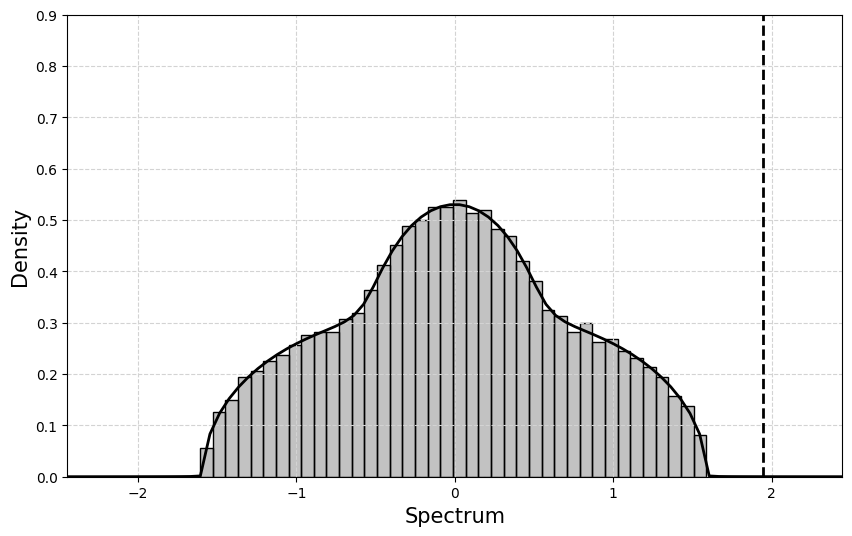}
    \caption{$\bs{\beta} = [3/4, 1/4]$, $\bs{s} = \begin{pmatrix}
1/3 & 1/5 \\ 
1 & 1/2
\end{pmatrix}$
}
\label{subfig:pattern_4}
\end{subfigure}
\caption{Spectrum (histogram) of the Hermitian random matrix $H$ ($n=1000$), conditions on $(\bs{\beta},\bs{s})$ are given in each panel. The numerical approach is used to compute the solid line spectrum distribution. An upper bound for the largest eigenvalue of $H$, given by $2\left \| S \right \|_\infty^{1/2}$, is denoted by the dashed vertical line.}
\label{fig:semi-circle}
\end{figure}

\section{Remaining computations}
\label{app:proof_heuristics}
\subsection{Moments of $\check{Z}_k$}

We compute hereafter the conditional variance of $\check{Z}_k=(B\bs{x}^*)_k$ with respect to $\bs{x}^*$. It should be noted that the assumption of independence between ($x^*_k$) and ($B_{kl}$) is a strong one, although it is likely to be true asymptotically, supported by the chaos hypothesis.

We rely on the following identities $\forall k \in \mathcal{I}_i$, $\forall \ell \in \mathcal{I}_j$ and $\forall o \in \mathcal{I}_q$:
$$
\mathbb{E} B_{k\ell} = 0 \quad,\quad \mathbb{E} (B_{k\ell}^2) = \frac {s_{ij}^2}{ n } \quad,\quad \mathbb{E} B_{k\ell}B_{ko} = 0 \quad (\ell \neq o)\, .
$$
We first compute the conditional mean:
\begin{align*}
\forall k \in \mathcal{I}_i, \ \mathbb{E}_{\bs{x}^*}(\check{Z}_k) &= \sum_{\ell\in[n]}  \mathbb{E}(B_{k\ell}) x_\ell^*
= \sum_{\ell\in {\mathcal S}_1}  \mathbb{E}(B_{k\ell}) x_\ell^*+\sum_{\ell\in {\mathcal S}_2}  \mathbb{E}(B_{k\ell}) x_\ell^* =  0\, .
\end{align*}
We now compute the second moment:
\begin{eqnarray*}
\forall k \in \mathcal{I}_i, \ \mathbb{E}_{\bs{x}^*}(\check{Z}_k^2) &=&\mathbb{E}_{\bs{x}^*} \left( \sum_{\ell\in [n]}  B_{k\ell} x_\ell^{*}\right)^2
    \ =\  \mathbb{E}_{\bs{x}^*}  \sum_{\ell,o\in [n]} B_{k\ell}B_{ko}x_\ell^*x_o^*\,,\\ 
    &=& \sum_{\ell \in \mathcal{S}_1 \cup \mathcal{S}_2}\mathbb{E}(B_{k\ell}^2)x_\ell^{*2}+\sum_{\ell \neq o}\mathbb{E} (B_{k\ell}B_{ko})x_\ell^*x_o^*, \\
    &= &\frac{s_{i1}^2}{ n} \sum_{\ell \in \mathcal{S}_1}x_\ell^{*2}+\frac{s_{i2}^2}{ n} \sum_{\ell \in \mathcal{S}_2}x_\ell^{*2} \, ,\\
    &=& \beta_1 \hat{p}_1 \hat{\sigma}^2_1s_{i1}^2+\beta_2 \hat{p}_2 \hat{\sigma}^2_2s_{i2}^2 \, , 
\end{eqnarray*}
We can now compute the variance:
$$
\forall  \, k \in \mathcal{I}_i, \ \textrm{Var}_{\bs{x}^*} \left(\check{Z}_k\right) =   \mathbb{E}_{\bs{x}^*}\left(\check{Z}_k^2\right) -  \left( \mathbb{E}_{\bs{x}^*}\check{Z}_k\right)^2 \ =\ \beta_1 \hat{p}_1 \hat{\sigma}^2_1s_{i1}^2+\beta_2 \hat{p}_2 \hat{\sigma}^2_2s_{i2}^2\, .
$$

\subsection{Details of heuristics of the mean}
\label{app:heur_mean}
Our starting point is the following generic representation of an abundance at equilibrium (either of a surviving or vanishing species) in the case $k \in \mathcal{S}_i$:
\begin{equation*}
x_k^* = \left( 1 +\Delta^*_i Z_k \right)\bs{1}_{\{Z_k>  \delta^*_i\}} = \bs{1}_{\{Z_k>  \delta^*_i\}} +\left( \Delta^*_i Z_k \right)\bs{1}_{\{Z_k>  \delta^*_i\}} \, .
\end{equation*}
Summing over ${\mathcal S_i}$ and normalizing, 
\begin{equation*}
\begin{split}
\frac{1}{|{\mathcal S}_i|}\sum_{k \in \mathcal{S}_i}x_k^{*} &=\frac{1}{|{\mathcal S}_i|}\sum_{k \in \mathcal{S}_i}\bs{1}_{\{Z_k > \delta^*_i\}}+\Delta^*_i\frac{1}{|{\mathcal S}_i|}\sum_{k \in \mathcal{S}_i}Z_k\bs{1}_{\{Z_k > \delta^*_i\}}, \\
\mhat_i &\stackrel{(a)}= 1+\Delta^*_i\frac{|\mathcal{I}_i|}{|{\mathcal S}_i|}\frac{1}{|\mathcal{I}_i|}\sum_{k \in \mathcal{I}_i}Z_k\bs{1}_{\{Z_k > \delta^*_i\}}, \\
\mhat_i &\stackrel{(b)}\simeq 1+\Delta^*_i \frac{1}{\mathbb{P}(Z > \delta^*_i)}\mathbb{E}(Z\bs{1}_{\{Z > \delta_i^*\}}), \\
\mhat_i &\simeq 1+\Delta^*_i\mathbb{E}(Z \mid Z > \delta^*_i).
\end{split}
\end{equation*}
where $(a)$ follows from the fact that $|{\mathcal S}_i| = \sum_{k \in \mathcal{S}_i}\bs{1}_{\{Z_k > \delta^*_i \}}$ (by definition of ${\mathcal S}_i$), $(b)$ from the law of large numbers $\frac{1}{|\mathcal{I}_i|} \sum_{k\in \mathcal{I}_i} Z_k \bs{1}_{\{Z_k>\delta_i\}} \xrightarrow[n\to\infty]{} \mathbb{E}Z \bs{1}_{\{Z>\delta^*_i\}}$ and $\frac{|{\mathcal S}_i|}{|\mathcal{I}_i|} \xrightarrow[n\to\infty]{} \mathbb{P}(Z>\delta^*_i)$ with $Z\sim{\mathcal N}(0,1)$. It remains to replace $\mhat_i$ by its limit $m^*_i$ to obtain the heuristics of the mean:
\begin{equation*}
      m^*_1 = 1+\Delta^*_1 \mathbb{E}(Z \mid Z>\delta^*_1)\, , 
\end{equation*}
\begin{equation*}
    m^*_2 = 1+\Delta^*_2 \mathbb{E}(Z \mid Z>\delta^*_2)\, .
\end{equation*}

\subsection{Density of the distribution of the surviving species.}
Assume that $x^*>0$, and let $f=\mathbb{R}\to \mathbb{R}$ be a bounded continuous test function, then $\forall k\in \mathcal{S}_i$
\begin{eqnarray*}
    \mathbb{E}f(x_k^*) &=& \mathbf{E}\left[f\left(1+\Delta^*_i Z_k \right)\  \bigg|\  Z_k > \delta^*_i \right]\ , \\
    &=& \int_{-\infty}^{\infty}f\left(1 +\Delta^*_i u\right) \frac{\bs{1}_{\{u>\delta_i^*\}}}{1-\Phi(\delta_i^*)}\frac{e^{-\frac{u^2}{2}}}{\sqrt{2\pi}}du\ , \\
    &=& \int_{0}^{\infty}f(y) e^{-\frac 12\left(\frac{y}{\Delta^*_i}+\delta_i^*\right)^2}\frac{1}{\sqrt{2\pi} \Phi(-\delta_i^*)\Delta^*_i}\,dy\ ,
\end{eqnarray*}
hence the density of $x^*_k, \, \forall k \ \in \mathcal{S}_i$ .

\section{Sketch of proof of Theorem \ref{th:fea_2b}}
\label{app:proof_feasibility_2b}
The first step consists in decomposing the equilibrium $\bs{x}^*$:
\begin{eqnarray*}
    x^*_k &=& e_k^{\top} \bs{x^*} \quad =\quad  e_k^{\top}(I-B)^{-1}\bs{1} \quad
    = \quad \sum_{\ell = 0}^{\infty} e_k^{\top} B^{\ell} \bs{1} \quad =\quad 
     1 + e_k^{\top}  B\bs{1}+e_k^{\top} B^2(I-B)^{-1} \bs{1} \, , \\
    &=& 1 + Z_k + R_k \, ,
\end{eqnarray*}
where $Z_k = \sum_{\ell=1}^n B_{k\ell} \ , \ \forall k \in [n]$.

One can prove that $\forall k \in [n], \, R_k$ is a negligible term if $n$ is sufficiently large. From a technical point, it relies on Gaussian concentration of Lipschitz functionnals and we are confident that the techniques applied in \cite{bizeul_positive_2021} will succeed in handling $R_k$. However, this part of the proof is not been treated here since we want to stick a concise argumentation of the proof which gives the reader information about the critical bound of the feasibility threshold.

The feasibility of the two communities is studied independently. Using Gaussian addition properties, a simpler form of $Z_k$ is first deduced. Consider a family $(\check{Z}_k)_{k\in[n]}$ of i.i.d. random variables $\mathcal{N}(0, 1)$. 
\begin{align*}
    \text{If }k\in \mathcal{I}_1, \ Z_k &= \sum_{\ell \in \mathcal{I}_1}B_{k\ell}+\sum_{\ell \in \mathcal{I}_2}B_{k\ell} \, , \\
    &\sim \mathcal{N}\left(0\,,\, \beta_1 s_{11}^2\right) + \mathcal{N}\left(0\,,\,\beta_2s_{12}^2\right) \, , \\
    &\sim \sqrt{\beta_1s_{11}^2+\beta_2s_{12}^2}\check{Z}_k\,.
\end{align*}
Similarly
\begin{equation*}
    \text{if }k\in \mathcal{I}_2, \ Z_k \sim \sqrt{\beta_1s_{21}^2+\beta_2s_{22}^2}\check{Z}_k\,.
\end{equation*}
Given $\bs{\beta} = (\beta_1,\beta_2)$, conditions on the matrix $\bs{s}$ are inferred to have:
\begin{equation}
    \mathbb{P}\left(\underset{k\in[n] }{\min} \ x_k > 0\right) = 1 \quad \Leftrightarrow \quad \mathbb{P}\left(\underset{k\in[n] }{\min} \ Z_k > -1\right) = 1\,.
    \label{eq:fea_min_x}
\end{equation}
In order to compute a tractable form of $\underset{k\in[n] }{\min} \ Z_k$, an additional approximation is made. Since  $(\check{Z}_k)_{k\in[n]}$ is a family of i.i.d. random variables $\mathcal{N}(0, 1)$, using standard extreme value theory of Gaussian random variables (see Leadbetter \textit{et al.} \cite{leadbetter_extremes_1983}), if $n$ is large enough 
\begin{equation}
\label{eq:approx_log}
\forall i \in \{1,2\}, \quad
\underset{k\in \mathcal{I}_i }{\min} \ \Check{Z}_k \sim - \sqrt{2\log(\beta_in)} \approx -\sqrt{2\log(n)}  \end{equation}

\begin{align*}
    \underset{k\in[n] }{\min} \ Z_k &= \min \left(\sqrt{\beta_1s_{11}^2+\beta_2s_{12}^2}\underset{k\in \mathcal{I}_1 }{\min} \ \check{Z}_k,\sqrt{\beta_1s_{21}^2+\beta_2 s_{22}^2}\underset{k\in \mathcal{I}_2 }{\min} \ \check{Z}_k \right)\, , \\
     & \stackrel{(a)}\simeq \min \left(\sqrt{\beta_1s_{11}^2+\beta_2s_{12}^2}\left(-\sqrt{2\log(n)}\right),\sqrt{\beta_1s_{21}^2+\beta_2 s_{22}^2}\left(-\sqrt{2\log(n)}\right)\right)\, , \\
    &= \min \left(-\sqrt{2\beta_1s_{11}^2\log(n)+2\beta_2s_{12}^2\log(n)},-\sqrt{2\beta_1s^2_{21}\log(n)+2\beta_2s_{22}^2\log(n)}\right)\, , \\
    &= - \max \left(\sqrt{2\beta_1s_{11}^2\log(n)+2\beta_2s_{12}^2\log(n)},\sqrt{2\beta_1s_{21}^2\log(n)+2\beta_2s_{22}^2\log(n)}\right)\, .
\end{align*}
where $(a)$ comes from the approximation \eqref{eq:approx_log}. The condition $\underset{k\in[n] }{\min} \ Z_k > -1$ asymptotically boils down to
\begin{align*}
&  \max \left(\sqrt{2\beta_1s_{11}^2\log(n)+2\beta_2s_{12}^2\log(n)},\sqrt{2\beta_1s_{21}^2\log(n)+2\beta_2s_{22}^2\log(n)}\right) < 1 \, , \\
\Leftrightarrow \quad &\max \left(2\beta_1s_{11}^2\log(n)+2\beta_2s_{12}^2\log(n),2\beta_1s_{21}^2\log(n)+2\beta_2s_{22}^2\log(n)\right) < 1 \, , \\
\Leftrightarrow \quad &\max \left(\beta_1s_{11}^2+\beta_2s_{12}^2,\beta_1s_{21}^2+\beta_2s_{22}^2\right) < \frac{1}{2\log(n)}\, , \\
\Leftrightarrow \quad &\left \| (\bs{s_n}\circ\bs{s_n}) \bs{\beta}^\top \right \|_{\infty} < \frac{1}{2\log(n)} := (s_n^*)^2\, .
\end{align*}

\section{Extension to the $b$-blocks model}
\label{app:extension_b_block}
\subsection{Interaction matrix with $b$ communities}
Within the framework of $b$ communities, the matrix $B = (B_{k\ell})_{n,n}$ is defined as

\begin{equation}
B = V\bs{s}V^{\top}\circ \frac{1}{\sqrt{n}}A \, ,
\label{eq:mat_int_extended}
\end{equation}
where
$$ V \in \mathcal{M}_{n\times b}, \,
V = \begin{pmatrix}
\bs{1}_{\mathcal{I}_1} & 0 & \cdots  & 0\\ 
0 & \bs{1}_{\mathcal{I}_2} & \cdots  & 0\\ 
\vdots  & \vdots  & \ddots  &\vdots  \\ 
0 & 0 & \cdots  & \bs{1}_{\mathcal{I}_b}
\end{pmatrix} \, , \, A = \begin{pmatrix}
A_{11} & \cdots  & A_{1b}\\ 
\vdots & \ddots  & \vdots   \\ 
A_{b1} & \cdots &  A_{bb}
\end{pmatrix} \, , \bs{s} = \begin{pmatrix}
s_{11} & \cdots  & s_{1b}\\ 
\vdots & \ddots  & \vdots   \\ 
s_{b1} & \cdots &  s_{bb}
\end{pmatrix} \, ,
$$
where $\mathcal{I}_1 = [n_1]$, ${\mathcal I}_2=\{n_1 +1,\cdots, n_1+n_2\}$,$\dots$, ${\mathcal I}_b=\{n_1+\dots+n_{b-1} +1,\cdots, n\}$  the subset of $[n]$ of size $|\mathcal{I}_i|:=n_i$ matching the index of species belonging to community $i$ and $\bs{\beta}  = (\beta_1,\beta_2,..,\beta_b)$ with $\forall i \in [b], \, \beta_i = n_i/n$ and $\sum_{i=1}^b \beta_i = 1$. 
The random matrix $A_{ij}$ is non-Hermitian of size $|\mathcal{I}_i|\times|\mathcal{I}_j|$ with standard Gaussian entries i.e. $\mathcal{N}(0,1)$. Recall that $\bs{1}_{\mathcal{I}_i}$ be the element-wise vector of $1$ with size $|\mathcal{I}_i|$.
\subsection{Existence of a unique equilibrium}
Let $H$ be the symmetric matrix
\begin{equation*}
    H = B+B^\top=\frac{1}{\sqrt{n}} \begin{pmatrix}
H_{11} & \cdots  & H_{1b}\\ 
\vdots  & \ddots  & \vdots  \\ 
H_{b1} & \cdots  & H_{bb} 
\end{pmatrix}\, ,
\end{equation*}
where $\forall \, i,j \in [b],\, H_{ij}$ is a matrix of size $|\mathcal{I}_i|\times|\mathcal{I}_j|$ and each off-diagonal entries follow a Gaussian distribution $\mathcal{N}(0,s_{ij}^2+s_{ji}^2)$.
The QVE associated to the matrix $H$ is decomposed as
\begin{equation*}
    k \in \mathcal{I}_i\, ,\, -\frac{1}{m_k(z)} = z + \sum_{j=1}^b \sum_{\ell \in \mathcal{I}_j}\frac{1}{n}\left(s_{ij}^2+s_{ji}^2\right)  m_\ell(z)\, .
\end{equation*}
Given $\bs{m}(z) = (m_1(z),\cdots,m_n(z))$, denote by $1/\bs{m}(z) = (1/m_1(z),\cdots,1/m_n(z))$  and $S = \frac{1}{n}V(\bs{s}+\bs{s}^{\top})V^{\top}$
the QVE can be written in the standard form
\begin{equation}
    -\frac{1}{\bs{m}(z)} = z+S\bs{m}(z)\, .
    \label{eq:b_block_QVE}
\end{equation}

Following the same arguments as in Theorem \ref{th:block_unicite_centered} and relying on Theorem \ref{th:block_takeuchi}, given the particular shape of the matrix $S$, computing its norm is equivalent to computing the norm of a matrix of size $b$
\begin{equation*}
\left \| S \right \|_\infty = \left \|\mathrm{diag}(\bs{\beta}) \left((\bs{s}\circ\bs{s})+(\bs{s}\circ\bs{s})^{\top}\right) \right \|_\infty \, .
\end{equation*}
\subsection{Surviving species}
\begin{heur}
\label{heur:properties_b_blocs}
Let $\bs{s}$ be the $b \times b$ matrix of interaction strengths and assume that the condition of Theorem \ref{th:block_unicite_centered} holds, then the following system of $2b$ equations and $2b$ unknowns $\bs{p} = (p_1,p_2,..,p_b), \, \bs{\sigma}=(\sigma_1,\sigma_2,..,\sigma_b)$
\begin{eqnarray*}
\forall i \in [b]\, ,\, 
p_i &=&1- \Phi(\delta_i) \ , \\
\forall i \in [b]\, ,\,
(\sigma_i)^2 &=&1+2\Delta_i \mathbb{E}(Z|Z>\delta_i)+\Delta_i^2\mathbb{E}(Z^2|Z>\delta_i) \ , 
\end{eqnarray*}
where
\begin{equation*}
    \Delta_i = \sqrt{\sum_{j=1}^b p_j (\sigma_j)^2\beta_j s_{ij}^2} \text{ ; } \delta_i = -\frac{1}{\Delta_i}, 
\end{equation*}
admits a unique solution $(\bs{p^*},\, \bs{\sigma^*})$ and $\forall \ i \in [b]$
 $$
 \hat{p}_i \xrightarrow[n\to\infty]{a.s.} p_i^* \qquad \text{and}\qquad \hat{\sigma}_i\xrightarrow[n\to\infty]{a.s.} \sigma^*_i\, .
$$
\end{heur}
\subsection{Distribution of the surviving species}
Let $\bs{s}$ be the $b \times b$ matrix of interaction strengths and assume that the condition of Theorem \ref{th:block_unicite_centered} holds. Let $\bs{x}^*$ the solution of \eqref{eq:block_equilibrium-NI} and $(\bs{p^*},\, \bs{\sigma^*})$ the solution of the heuristic \ref{heur:properties_b_blocs}. Recall the definition of $\Delta_i,\, \delta_i$ and denote by $\delta_i^* = \delta_i(p_i^*,\sigma_i^*)$. Let $x^*_k>0$ a positive component of $\bs{x}^*$ belonging to the community $i$, then:
$$
\mathcal{L}(x_k^*) \xrightarrow[n\to\infty]{} \mathcal{L}\left(1+\Delta^*_i Z \quad \bigg|\quad Z > \delta_i^* \right)\ ,
$$
where $Z\sim {\mathcal N}(0,1)$. Otherwise stated, asymptotically $ \forall k \in \mathcal{S}_i, \, x^*_k$ admits the following density
\begin{equation}
 f_k(y) = \frac{\bs{1}_{\{y>0\}}}{\Phi(-\delta^*_i)}\frac{1}{\Delta^*_i\sqrt{2\pi}}\,
 \exp\left\{-\frac 12\left(\frac{y}{\Delta^*_i}+\delta^*_i\right)^2\right\}\ .
\end{equation}

\subsection{Feasibility}
We consider a growing scaling matrix 
$$
\bs{s}_n \xrightarrow[n\to\infty]{} \bs{0} \quad \Leftrightarrow \quad \forall i,j \in \{1,b\}, s_{ij} \xrightarrow[n\to\infty]{} 0 \, .
$$
Let $B_n$ a matrix defined by
\begin{equation}
B_n = V\bs{s}_nV^{\top}\circ\frac{1}{\sqrt{n}}A \, .
\label{eq:mat_int_extended_fea}
\end{equation}
The spectral radius of $\frac{1}{\sqrt{n}}A$ a.s. converges to $1$ (circular law). So as long as $\bs{s}_n$ is close to zero, the matrix $I - B_n$ is eventually invertible.

Recall the problem which admits a unique solution defined by
\begin{equation}
    \bs{x^*} = \bs{1}+B\bs{x^*} \ \Leftrightarrow \ \bs{x^*} = (I-B)^{-1} \bs{1}\, ,
    \label{eq:fea_fixed_b_blocks}
\end{equation}

\begin{theo}[Co-feasibility for the $b$-blocks model]
Assume that matrix $B_n$ is defined by the $b$-blocks model \eqref{eq:mat_int_extended_fea}. Let $\bs{\beta} = (\beta_1,\beta_2,..,\beta_b), \, \sum_{i=1}^b \beta_i = 1$ represents the proportion of each community. Let $\bs{s}_n \xrightarrow[n\to\infty]{} 0$ and denote by $s_n^*=1/\sqrt{2\log n}$. Let $x_n = (x_k)_{k\in[n]}$ be the
solution of \eqref{eq:fea_fixed_b_blocks}.
\begin{enumerate}
    \item If there exists $\varepsilon > 0$ such that eventually $\left \| (\bs{s_n}\circ \bs{s_n}) \bs{\beta}^\top \right \|_{\infty} \geq (1+\varepsilon)(s_n^*)^2$ then
    $$
    \mathbb{P}\left \{ \underset{k\in [n]}{\min}\, x_k>0 \right \} \xrightarrow[n\rightarrow \infty]{} 0 \, .
    $$
    \item If there exists $\varepsilon > 0$ such that eventually $\left \| (\bs{s_n}\circ \bs{s_n}) \bs{\beta}^\top \right \|_{\infty} \leq (1-\varepsilon)(s_n^*)^2$ then
    $$
    \mathbb{P}\left \{ \underset{k\in [n]}{\min}\, x_k>0 \right \} \xrightarrow[n\rightarrow \infty]{} 1 \, .
    $$
\end{enumerate}
\end{theo}
\begin{proof}[Sketch of proof]
Starting from the decomposition, the equilibrium $\bs{x}^*$:
$$
x^*_k = 1 + Z_k + R_k \, ,
$$
where $Z_k = \sum_{\ell=1}^n B_{k\ell} \ , \ \forall \, k  \in [n]$ \ and we assume that $\forall \, k \in [n]\,, \, R_k$ is a negligible term if $n$ is sufficiently large.

The feasibility of the $b$ communities is studied independently. Using Gaussian addition properties, a simpler form of $Z_k$ is derived. Consider a family $(\check{Z}_k)_{k\in[n]}$ of i.i.d. random variables $\mathcal{N}(0, 1)$. 
$$
\text{If }k\in \mathcal{I}_i,\quad Z_k \quad =\quad  \sum_{j=1}^b \sum_{\ell \in \mathcal{I}_j}B_{k\ell} \quad 
    \sim \quad \sum_{j=1}^b \mathcal{N}\left(0\,,\, \beta_js_{ij}^2\right) 
    \quad \sim\quad  \sqrt{\sum_{j=1}^b \beta_j s_{ij}^2}\check{Z}_k\,.
$$
Given $\bs{\beta} = (\beta_1,\beta_2,..,\beta_b)$, conditions on the matrix $\bs{s}$ are inferred to have
\begin{equation*}
    \mathbb{P}(\underset{k\in[n] }{\min} \ x_k > 0) = 1 \quad \Leftrightarrow\quad  \mathbb{P}(\underset{k\in[n] }{\min} \ Z_k > -1) = 1\,.
    \label{eq:fea_min_x_multi}
\end{equation*}
In order to compute a tractable form of $\underset{k\in[n] }{\min} \ Z_k$, an additional approximation is made, if $n$ is large enough 
\begin{equation}
\label{eq:approx_log_multi}
\underset{k\in \mathcal{I}_i }{\min} \ \Check{Z}_k \quad \sim\quad  - \sqrt{2\log(\beta_in)}\quad  \simeq\quad  -\sqrt{2\log(n)} \, .   
\end{equation}
With this approximation at hand, we can proceed:
\begin{eqnarray*}
    \underset{k\in[n] }{\min} \ Z_k \quad =\quad  \underset{i\in[b]}{\min} \left( \sqrt{\sum_{j=1}^b \beta_j s_{ij}^2}\underset{k\in \mathcal{I}_i }{\min} \ \check{Z}_k\right)
    &\simeq &\underset{i\in[b]}{\min} \left( \sqrt{\sum_{j=1}^b \beta_j s_{ij}^2}\left(-\sqrt{2\log(n)}\right)\right)\\
    &=&  \underset{i\in[b]}{\min} \left(- \sqrt{\sum_{j=1}^b 2\beta_js_{ij}^2\log(n)}\right)
    \quad =\quad  -\underset{i\in[b]}{\max} \left( \sqrt{\sum_{j=1}^b 2\beta_j s_{ij}^2\log(n)}\right)\, .
\end{eqnarray*}
Following the approximation \eqref{eq:approx_log_multi}, the condition $\underset{k\in[n] }{\min} \ Z_k > -1$ asymptotically boils down to
\begin{eqnarray*}
\underset{i\in[b]}{\max} \left( \sqrt{\sum_{j=1}^b 2\beta_j s_{ij}^2\log(n)}\right) < 1 
    &\Leftrightarrow  & \underset{i\in[b]}{\max} \left( \sum_{j=1}^b 2\beta_j s_{ij}^2 \log(n)\right) < 1 \, , \\
    &\Leftrightarrow  & \underset{i\in[b]}{\max} \left( \sum_{j=1}^b \beta_j s_{ij}^2\right) < \frac{1}{2\log(n)} 
    \quad \Leftrightarrow \quad  \left \| (\bs{s_n}\circ \bs{s_n})^2\bs{\beta}^\top \right \|_{\infty} < \frac{1}{2\log(n)} := (s_n^*)^2 \, .
\end{eqnarray*}

\end{proof}

\end{document}